\documentclass{article}

\usepackage{arxiv}
\usepackage{amsthm}
\usepackage[utf8]{inputenc} % allow utf-8 input
\usepackage[T1]{fontenc}    % use 8-bit T1 fonts
\usepackage{hyperref}       % hyperlinks
\usepackage{url}            % simple URL typesetting
\usepackage{booktabs}       % professional-quality tables
\usepackage{amsfonts}       % blackboard math symbols
\usepackage{nicefrac}       % compact symbols for 1/2, etc.
\usepackage{microtype}      % microtypography
\usepackage{graphicx}
\usepackage{natbib}
\usepackage{doi}
\usepackage{amsmath}
\usepackage{amssymb} 

\usepackage[linesnumbered,ruled,vlined]{algorithm2e}
\usepackage{multirow}

\SetKwInput{KwInput}{Input}                % Set the Input
\SetKwInput{KwOutput}{Output}
\SetKwComment{Comment}{/* }{ */}
\SetKwRepeat{Do}{do}{while}

\DeclareMathOperator{\E}{\mathbb{E}}
\DeclareMathOperator{\Var}{\mathrm{Var}}
\DeclareMathOperator{\Cov}{\mathrm{Cov}}

\DeclareSymbolFont{boperators}   {OT1}{cmr} {bx}{n}
\DeclareSymbolFont{bletters}     {OML}{cmm} {b}{it}
\DeclareSymbolFont{bsymbols}     {OMS}{cmsy}{b}{n}

\DeclareMathSymbol{\BFa}{\mathalpha}{boperators}{`a}
\DeclareMathSymbol{\BFb}{\mathalpha}{boperators}{`b}
\DeclareMathSymbol{\BFc}{\mathalpha}{boperators}{`c}
\DeclareMathSymbol{\BFd}{\mathalpha}{boperators}{`d}
\DeclareMathSymbol{\BFe}{\mathalpha}{boperators}{`e}
\DeclareMathSymbol{\BFf}{\mathalpha}{boperators}{`f}
\DeclareMathSymbol{\BFg}{\mathalpha}{boperators}{`g}
\DeclareMathSymbol{\BFh}{\mathalpha}{boperators}{`h}
\DeclareMathSymbol{\BFi}{\mathalpha}{boperators}{`i}
\DeclareMathSymbol{\BFj}{\mathalpha}{boperators}{`j}
\DeclareMathSymbol{\BFk}{\mathalpha}{boperators}{`k}
\DeclareMathSymbol{\BFl}{\mathalpha}{boperators}{`l}
\DeclareMathSymbol{\BFm}{\mathalpha}{boperators}{`m}
\DeclareMathSymbol{\BFn}{\mathalpha}{boperators}{`n}
\DeclareMathSymbol{\BFo}{\mathalpha}{boperators}{`o}
\DeclareMathSymbol{\BFp}{\mathalpha}{boperators}{`p}
\DeclareMathSymbol{\BFq}{\mathalpha}{boperators}{`q}
\DeclareMathSymbol{\BFr}{\mathalpha}{boperators}{`r}
\DeclareMathSymbol{\BFs}{\mathalpha}{boperators}{`s}
\DeclareMathSymbol{\BFt}{\mathalpha}{boperators}{`t}
\DeclareMathSymbol{\BFu}{\mathalpha}{boperators}{`u}
\DeclareMathSymbol{\BFv}{\mathalpha}{boperators}{`v}
\DeclareMathSymbol{\BFw}{\mathalpha}{boperators}{`w}
\DeclareMathSymbol{\BFx}{\mathalpha}{boperators}{`x}
\DeclareMathSymbol{\BFy}{\mathalpha}{boperators}{`y}
\DeclareMathSymbol{\BFz}{\mathalpha}{boperators}{`z}
\DeclareMathSymbol{\BFA}{\mathalpha}{boperators}{`A}
\DeclareMathSymbol{\BFB}{\mathalpha}{boperators}{`B}
\DeclareMathSymbol{\BFC}{\mathalpha}{boperators}{`C}
\DeclareMathSymbol{\BFD}{\mathalpha}{boperators}{`D}
\DeclareMathSymbol{\BFE}{\mathalpha}{boperators}{`E}
\DeclareMathSymbol{\BFF}{\mathalpha}{boperators}{`F}
\DeclareMathSymbol{\BFG}{\mathalpha}{boperators}{`G}
\DeclareMathSymbol{\BFH}{\mathalpha}{boperators}{`H}
\DeclareMathSymbol{\BFI}{\mathalpha}{boperators}{`I}
\DeclareMathSymbol{\BFJ}{\mathalpha}{boperators}{`J}
\DeclareMathSymbol{\BFK}{\mathalpha}{boperators}{`K}
\DeclareMathSymbol{\BFL}{\mathalpha}{boperators}{`L}
\DeclareMathSymbol{\BFM}{\mathalpha}{boperators}{`M}
\DeclareMathSymbol{\BFN}{\mathalpha}{boperators}{`N}
\DeclareMathSymbol{\BFO}{\mathalpha}{boperators}{`O}
\DeclareMathSymbol{\BFP}{\mathalpha}{boperators}{`P}
\DeclareMathSymbol{\BFQ}{\mathalpha}{boperators}{`Q}
\DeclareMathSymbol{\BFR}{\mathalpha}{boperators}{`R}
\DeclareMathSymbol{\BFS}{\mathalpha}{boperators}{`S}
\DeclareMathSymbol{\BFT}{\mathalpha}{boperators}{`T}
\DeclareMathSymbol{\BFU}{\mathalpha}{boperators}{`U}
\DeclareMathSymbol{\BFV}{\mathalpha}{boperators}{`V}
\DeclareMathSymbol{\BFW}{\mathalpha}{boperators}{`W}
\DeclareMathSymbol{\BFX}{\mathalpha}{boperators}{`X}
\DeclareMathSymbol{\BFY}{\mathalpha}{boperators}{`Y}
\DeclareMathSymbol{\BFZ}{\mathalpha}{boperators}{`Z}

\newtheorem{definition}{Definition}
\newtheorem{assumption}{Assumption}

\newtheorem{theorem}{Theorem}
\newtheorem{corollary}{Corollary}
\newtheorem{lemma}{Lemma}

\title{Causal Inference under Network Interference Using a Mixture of Randomized Experiments}

\author{ {\hspace{1mm}Yiming Jiang} \\
	Georgia Institute of Technology\\
	Atlanta, GA 30318 \\
	\texttt{yjiang463@gatech.edu}\\
	\And
	{\hspace{1mm}He Wang} \\
	Georgia Institute of Technology\\
	Atlanta, GA 30318 \\
	\texttt{he.wang@isye.gatech.edu} 
}

\newcommand{\undertitle}{}

\hypersetup{
pdfauthor={Yiming Jiang, He Wang},
pdfkeywords={experimental design, causal inference, network interference, clustering, SUTVA},
}

\begin{document}
\maketitle

\begin{abstract}
In randomized experiments, the classic \emph{Stable Unit Treatment Value Assumption} (SUTVA) posits that the outcome for one experimental unit is unaffected by the treatment assignments of other units. However, this assumption is frequently violated in settings such as online marketplaces and social networks, where interference between units is common. We address the estimation of the total treatment effect in a network interference model by employing a mixed randomization design that combines two widely used experimental methods: \emph{Bernoulli randomization}, where treatment is assigned independently to each unit, and \emph{cluster-based randomization}, where treatment is assigned at the aggregate level. The mixed randomization design simultaneously incorporates both methods, {thereby mitigating the bias present in cluster-based designs.} We propose an unbiased estimator for the total treatment effect under this mixed design and show that its variance is bounded by $O(d^2 n^{-1} p^{-1} (1-p)^{-1})$, where $d$ is the maximum degree of the network, $n$ is the network size, and $p$ is the treatment probability. Additionally, we establish a lower bound of $\Omega(d^{1.5} n^{-1} p^{-1} (1-p)^{-1})$ for the variance of any mixed design. Moreover, when the interference weights on the network's edges are unknown, we propose a \emph{weight-invariant design} that achieves a variance bound of $O(d^3 n^{-1} p^{-1} (1-p)^{-1})$, which is aligned with the estimator introduced by \cite{cortez2023low-order} under similar conditions.
\end{abstract}

% keywords can be removed

\section{Introduction}
Randomized experiments are a powerful tool for understanding the causal impact of changes. For example, online marketplaces use randomized experiments to test the effectiveness of new features \citep{johari2022twosided}; tech companies build large-scale experimentation platforms to improve their products and systems \citep{paluck2016socialexp, saveski2017detecting};  economists and social science researchers rely heavily on randomized experiments to understand the effects of economic and social changes \citep{leung2022ANI}. The overarching goal of randomized experiments is to impute the difference between two universes that cannot be observed simultaneously: a factual universe where the treatments assigned to all units remain unchanged, and a counterfactual universe where all units receive a new treatment. The difference in the average outcome between these two universes is commonly known as the {total treatment effect (TTE).}
Standard estimation approaches for the TTE such as A/B testing rely heavily on a fundamental independence assumption: the \emph{stable unit treatment value assumption} (SUTVA) \citep{imbens2015causal}, which states that the outcome for each experimental unit depends only on the treatment it received, and not on the treatment assigned of any other units. 

However, SUTVA is often violated in settings where experimental units interact with each other. In such cases, ignoring the interference between treatment and control groups may lead to significant estimation bias of the TTE.
For example, suppose an e-commerce retailer wants to determine the causal effect on its sales by applying a price promotion to all products.
One simple approach is to run an experiment that applies promotion randomly to a subset of products, and then compare the average sales of promoted products with non-promoted products. 
However, simply comparing the difference in the two groups will likely overestimate the total treatment effect of the promotion, because customers may alter their shopping behavior in the experiment by substituting non-promoted products for similar promoted products. As such, if the retailer decides to implement the price promotion for all products on its platform based on the result of this experiment, the realized sales lift may be smaller than expected.

One common approach for reducing the estimation bias of TTE in the presence of 
interference is \emph{cluster-based randomization}, which assigns treatment at an aggregate level rather than at the individual level. In such a design, the population of the experiment is often represented by an underlying network $G(V,E)$, where each vertex in the set $V$ represents an individual experimental unit, and a directed edge $(i,j) \in E$ indicates that the outcome of unit $i$ may be affected by the treatment of another unit $j$. The weights on the edges represent the magnitude of interference between two units. Note that the weights on $(i,j)$ and $(j,i)$ can be different, as the interference between two units may be asymmetric.
A cluster-based randomization design
partitions the network into disjoint subsets of vertices (i.e., clusters) and applies the same treatment to all units within the same cluster.
Clearly, using larger clusters will reduce the number of edges across clusters 
and will better capture the network interference effect between units, which leads to a smaller bias. However, a downside of using larger clusters is that there will be fewer \emph{randomization units} in the experiment, resulting in higher variance in the TTE estimation. 
Considering this bias-variance trade-off, a substantial body of literature delves into optimal cluster design under network interference
\citep{ugander2013graph, candogan2021robustcorrelated, leung2022spatialcluster, brennan2022clusterbipartite}. 

In this paper, we propose a modification of cluster-based randomization by mixing it with another commonly used experimental design in which each unit is individually and independently randomized, i.e., \emph{Bernoulli randomization}. {The key idea of this approach is that the \emph{difference-in-means} estimator under Bernoulli randomization provides an estimation of the direct treatment effect, while cluster-based randomization estimates the direct treatment effect plus a partial spillover effect. By combining these two randomization methods, we are able to ``learn" the full spillover effect and create an unbiased estimator of the total treatment effect (TTE).} As one experimental unit cannot receive two different treatments simultaneously, the mixed design assigns a randomization method to each unit based on a coin flip. Essentially, the mixed design \emph{simultaneously} runs a cluster-based randomized experiment on one part of the network and an individually randomized experiment on the other part. The idea of mixing cluster-based and Bernoulli designs has been explored by \citep{saveski2017detecting, pouget2019hierachicaltesting}, who use the mixed design to detect the existence of network interference (i.e., hypothesis testing for SUTVA). However, to the best of our knowledge, the effect of using mixed randomization designs for estimating total treatment effects has not been previously studied.

We summarize the main results of the paper as follows. {In Section~\ref{section: set up}, we present the model and assumptions used in this article and introduce cluster-based randomization as a preliminary. We then explain how Bernoulli randomization can eliminate the bias in the cluster-based design and present our mixed randomization design.} In Section~\ref{section: known interference}, we define an unbiased estimator for TTE under mixed randomization designs and establish variance bounds for the estimator given any cluster design. Since these variance bounds cannot be directly applied to optimize cluster design due to computational challenges, we propose a heuristic for clustering and analyze the estimator's variance under this heuristic. For general networks, we prove that the TTE estimation variance is bounded by $O\left({d^2}/({np(1-p)})\right)$, where $d$ is the maximum degree of the network, $n$ is the network size, and $p$ is the fraction of units receiving treatment. We also establish a lower bound of $\Omega\left({d^{1.5}}/({np(1-p)})\right)$ for any mixed design, matching the upper bound up to a factor of $\sqrt{d}$. {Finally, we discuss why the commonly used regression adjustment approach is not appropriate in our setting.}

In Section~\ref{sec: unknown interfernece}, we extend the analysis to the case where the degree of interference between units, i.e., the edge weights of the network, are unknown \emph{a priori}. We propose a weight-invariant mixed randomization design that is agnostic to edge weights and show that this algorithm achieves an upper bound on the estimation variance of $O\left({d^3}/({np(1-p)})\right)$. {Our bound aligns with the variance bound of the pseudo-inverse estimator from \citep{cortez2023low-order, eichhorn2024pseudo-inverse} under the same linear model and improves upon the $O\left({d^6}/({np(1-p)})\right)$ upper bound from \citep{ugander2023randomized}, which uses a purely cluster-based randomization design.} Furthermore, we demonstrate that our proposed estimators are consistent and asymptotically normal, provided the graph is sufficiently sparse (Section~\ref{section: Inference}).

Finally, we present numerical experiments in Section~\ref{sec:simulation} to illustrate the efficacy of the proposed mixed randomization algorithms. We compare the performance of our method against cluster-based randomization and other alternative approaches.

\section{Related Works}

{There is extensive literature exploring causal inference in the presence of interference. For example, in sequential experiments, treatments applied at one time point can influence units in subsequent time slots through the carryover effect \citep{farias2022markovian,bojinov2023switchback}. Spatial interference often arises when treatments in certain regions spill over to nearby areas \citep{leung2022ANI}. In social platform experiments, treatments applied to individuals may affect the behavior of their friends \citep{cortez2023low-order,chen2024optimized-linear-model}. Similarly, in online markets, treatments on one unit may cannibalize the market share of another \citep{johari2022twosided}. Our work focuses on causal inference under network interference. Although we concentrate on general networks, there are studies on bipartite networks \citep{brennan2022clusterbipartite, harshaw2023bipartitelinear} and random networks \citep{li2022random-asymptotics}. A comprehensive review of interference is beyond the scope of this paper, but we recommend the review article by \cite{larsen2024ABT-review} for further information.}

{To address network interference, several experimental design approaches have been developed.} One common approach is \emph{cluster-based randomization}, where units are partitioned into clusters, and treatments are assigned at the cluster level. {Both theoretical and empirical studies have demonstrated that cluster-based designs can reduce bias in the presence of interference \citep{eckles2016reducebias,holtz2024empirical-cbr}. However, larger clusters reduce bias at the cost of higher variance, motivating research into clustering algorithms to minimize mean squared error \citep{viviano2023causal-clustering, ugander2023randomized}. The mixed randomization design, which combines cluster-based and Bernoulli randomization, was first proposed by \cite{saveski2017detecting} for detecting network interference. Our work extends this mixed design to mitigate the bias from cluster-based randomization and produce an unbiased estimator. When treatments can be applied sequentially, \cite{cortez2022rollout-network-interference} demonstrated that staggered rollout designs provide unbiased TTE estimators under unknown network effects, which was later generalized by \cite{cortez2024combine-rollout-cluster} to combine clustering algorithms and reduce variance. \cite{viviano2020experimental} introduced a two-wave experimental design, where a first-wave (pilot) experiment is used to estimate and minimize variance.}

{Several post-experiment analysis approaches have also been explored in prior works. \cite{chin2019regression} studied the \emph{regression adjustment} approach, which leverages network structures to construct covariates that improve estimation precision. The \emph{exposure mapping} framework, introduced by \cite{aronow2017estimating}, assumes that potential outcomes are a function of a low-dimensional representation (exposure mapping) of the treatment assignment vector and covariates. Under known exposure mappings, \cite{ugander2013graph} showed that the Horvitz-Thompson estimator provides unbiased TTE estimates with variance of order $O(d^2 /(np^d))$, where $n$ is the population size, $d$ is the maximum degree of the network, and $p$ is the treatment probability. \cite{ugander2023randomized} later improved this result to $O(d^6 /(np))$. Cases with approximate or misspecified exposure mappings were also studied by \cite{leung2022ANI} and \cite{savje2023misspecified}, respectively. When the potential outcome function has a $\beta$-order polynomial representation, \cite{cortez2023low-order} proposed the SNIPE estimator, which guarantees unbiased TTE estimation under Bernoulli randomization with variance bounded by $O(d^{\beta+2}/(np^{\beta}))$. This result aligns with our analysis for the mixed estimator under a linear potential outcome model. \cite{eichhorn2024pseudo-inverse} renamed the SNIPE estimator as the \emph{pseudo-inverse estimator}, demonstrating that it has similar theoretical properties under cluster-based randomization.}

The main goal of this paper is to correctly estimate the total treatment effect (TTE) under network interference. In addition to this, other research goals have garnered attention. {Testing whether SUTVA holds is crucial for experimenters, and mixed randomization design, exact p-value methods, and increasing allocation approaches have been developed by \cite{pouget2019hierachicaltesting}, \cite{athey2018exact-p-value}, and \cite{han2023detecting-interference}, respectively, to address this.} \cite{savje2021average} considered the inferential target of estimating the direct treatment effect under unknown interference. \cite{candogan2021robustcorrelated} developed a family of experiments called independent block randomization, which aims to correlate the treatment probability of clusters to minimize worst-case variance.

\section{Problem Setup}\label{section: set up}

Let $[n] :=\{1,...,n\}$ for any $n\in\mathbb{Z}_+$. Throughout the paper, we use boldface symbols (e.g., $\BFx$) to denote a vector, $x_i$ to denote its $i^{\text{th}}$ element, and $\bar x$ to denote its mean.

\subsection{Model}
Below, we define a causal inference problem under network interference. Consider an experiment with $n$ individual units. Each unit $i \in [n]$ is associated with a potential outcome function $Y_i(\cdot):\{0,1\}^n\rightarrow \mathbb{R}$, which maps a treatment vector $\BFz \in \{0, 1\}^n$ to a potential outcome value. By convention, we assume the randomness only comes from the treatment assignment. Let the average expected outcome of all units given the assignment $\BFz$ be
\begin{align*}
    \mu(\BFz)=\frac{1}{n}\sum_{i=1}^n Y_i(\BFz).
\end{align*}
Let $\textbf{1} \in \mathbb{R}_n$ be a vector of 1's and $\textbf{0} \in \mathbb{R}_n$ be a vector of 0's.
Our goal is to measure the total treatment effect (TTE): $\E[\mu(\textbf{1})-\mu(\textbf{0})]$. 

Because the number of possible treatment assignments is exponential in $n$ (i.e., $2^n$), making causal inferences on the TTE is impractical unless further assumptions are imposed on the structure of the outcome function $Y_i$. Throughout this paper, we assume that $Y_i(\BFz)$ depends only on the treatments of units from a subset $\mathcal{N}_i \subset V$, which is referred to as the \emph{neighborhood set} of 
unit $i$. More formally, for any two assignment vectors $\BFz$ and $\BFz'$ such that $z_i=z'_i$ and $z_j=z'_j$ for all $j\in \mathcal{N}_i$, we have $Y_i(\BFz)=Y_i(\BFz')$. We assume that the neighborhood set is known and correctly specified for each unit. By connecting units with their neighbors, we get a directed graph $G(V, E)$ where $V$ is the set of all units and $E$ is the set of all edges. Without loss of generality, we assume the graph is connected; otherwise, the problem can be decomposed into separate subgraphs. 

We focus on a linear model where the magnitude of interference on unit $i$ from unit $j$ is measured by a constant weight $v_{ij}$, $\forall (i,j) \in E$, which is summarized by the following assumption:
\begin{assumption}\label{ass: linear model}
    We assume that
the outcome function $Y_i(\BFz)$ of unit $i$ is given by
\begin{equation}
\label{eq:interference model}
Y_i(\BFz)=\alpha_i+\beta_iz_i+\gamma\sum_{j\in \mathcal{N}_i}v_{ij}z_j,\quad\forall i\in V,
\end{equation}
where $\alpha_i$, $\beta_i$, $\gamma$, and $v_{ij}$'s are fixed constants, $\sum_{j\in \mathcal{N}_i}|v_{ij}|\le 1,\;\forall i\in V$ and $\sum_i\sum_{j\in \mathcal{N}_i} v_{ij}\ge 0$.
\end{assumption}

In Eq~\eqref{eq:interference model}, $\alpha_i\in \mathbb{R}$ is the potential outcome of unit $i$ without any treatment, $\beta_i\in \mathbb{R}$ is the direct treatment effect on unit $i$, and $\gamma v_{ij} z_j$ is the spillover effect from an adjacent unit $j \in \mathcal{N}_i$. Because $\sum_{j \in \mathcal{N}_i}|v_{ij}|\leq 1$, the value $|\gamma|$ can be interpreted as the maximum absolute interference effect from neighbors. We treat the coefficients $\alpha_i$, $\beta_i$, and $\gamma$ as fixed but unknown so that the only randomness in the observation
model is due to the choice of treatment assignment $\BFz$. In Section~\ref{section: known interference}, we assume the weights $v_{ij}$'s are known. For example, \cite{saveski2017detecting} considers a social network experiment where each unit represents a user, and $v_{ij}$ is the same for all $(i,j) \in E$ (up to a normalization factor). As another example, in a transportation network experiment, $v_{ij}$ may represent the traffic flow from node $j$ to node $i$. The setting with unknown $v_{ij}$'s will be studied in Section~\ref{sec: unknown interfernece}. 

The model in Eq~\eqref{eq:interference model} can be viewed as a first-order approximation of more general network interference patterns.
Similar linearity assumptions are common in the causal inference literature \citep{brennan2022clusterbipartite, harshaw2023bipartitelinear}, and our model generalizes some previous studies by allowing weighted interference from neighboring units. {Throughout the paper, we maintain the standard assumption that potential outcomes are uniformly bounded:
\begin{assumption}\label{ass: bounded outcomes}
    There exist positive constants $Y_L$ and $Y_M$ such that $Y_L\le Y_i(\BFz)\le Y_M$ for all $i \in V$ and all $\BFz \in \{0,1\}^n$. 
\end{assumption}
This assumption is without loss of generality since we can always shift the outcome of all units by a fixed constant.}
{
\subsection{Cluster-based Design, Estimator and Bias}
}
One common approach used in the literature to reduce the bias of TTE estimation under network interference is the \emph{cluster-based} randomization, where the network is partitioned into disjoint clusters, and units within the same cluster are assigned the same treatment. Specifically, let $m \in \mathbb{Z}_+$ and let the clusters $\{C_1,C_2,...,C_m\}$ be a partition of the vertex set $V$. We define $c(i)$ as a function that maps the index of unit $i$ to its associated cluster, i.e., $c(i)=t$ if and only if $i\in C_t$. The standard estimator for TTE in cluster-based design is the {\emph{difference-in-means} estimator}, which is defined below:

\begin{definition}[Cluster-based Randomization Experiment]
\label{def:cbr}
Let $p\in(0,1)$ be the treatment probability. Under cluster-based randomized design, each unit's treatment is a Bernoulli random variable, i.e., $z_i\sim$ Bernoulli($p$) for all $i \in V$. Moreover, for $i\ne j$, the correlation coefficient of $z_i$ and $z_j$ is $\text{corr}(z_i,z_j)=\boldsymbol{1}\{c(i)=c(j)\}$. The difference-in-means estimator for the TTE in cluster-based design (denoted by the subscript \emph{cb}) is
\begin{equation}\label{eq: cbr formula}
\hat\tau_{cb}=\frac{1}{n}\sum_{i=1}^n \left(\frac{z_i}{p}-\frac{1-z_i}{1-p}\right)Y_i(\BFz).
\end{equation}
\end{definition}

{The estimator represents the difference in the expected average between the treatment and control groups. The next lemma shows that the bias of $\hat\tau_{cb}$ arises from imperfect graph clustering:}

\begin{lemma}\label{lem: cluster-based expectation}
Under Assumption \ref{ass: linear model} and the cluster-based design in Definition \ref{def:cbr}, the estimator $\hat\tau_{cb}$ defined in Eq~\eqref{eq: cbr formula} has bias
\begin{align}
\text{TTE}-\E[\hat\tau_{cb}]=\frac{\gamma}{n}\sum_{i=1}^n\sum_{j\in \mathcal{N}_i}v_{ij}\boldsymbol{1}\{c(i)\ne c(j)\}.
\end{align}  
\end{lemma} 

\begin{proof}
Let $t_i={z_i}/{p} - {(1-z_i)}/{(1-p)}$. Since $\E[t_i z_j] = \boldsymbol{1}\{c(i)=c(j)\}$, the result follows by substituting Eq \eqref{eq:interference model} into Eq \eqref{eq: cbr formula} and taking the expectation.
\end{proof}

{By Lemma~\ref{lem: cluster-based expectation}, the bias is the sum of network effects from the between-cluster connections.} The estimator $\hat\tau_{cb}$ of a cluster-based design is unbiased if and only if all clusters are independent, i.e., there do not exist $i,j \in V$ such that $j\in \mathcal{N}_i$ but $c(i)\ne c(j)$. Nevertheless, in most randomized experiments, the underlying network cannot be perfectly partitioned into independent clusters. {We note that there exist unbiased Horvitz-Thompson (HT) estimators in cluster-based designs \cite[e.g.][]{ugander2013graph} under the exposure mapping framework, but the difference-in-means estimator defined in Eq~\eqref{eq: cbr formula} is commonly used in practice since the variance of the HT estimator is usually prohibitively high.}

{\subsection{Incorporate Bernoulli Randomization to Eliminate Bias}}

{To address the as introduced by imperfect graph partitioning, we propose a refined version of the cluster-based design called the mixed randomization design, which incorporates information from a Bernoulli randomization experiment to construct an unbiased estimator for the total treatment effect (TTE).}

As an example, consider a $d$-regular network (i.e., $|\mathcal{N}_i|=d$ for every $i \in [n]$) with $v_{ij}=1/d$ for all $(i,j) \in E$. In this case, the TTE is $\bar\beta + \gamma$, where $\bar\beta = \sum_{i=1}^n \beta_i / n$. However, assuming between-cluster connections account for half of the total connections, the estimator $\hat\tau_{cb}$ gives $\E[\hat\tau_{cb}] = \bar\beta + \gamma/2$, which introduces significant bias. To reduce this bias, we can additionally use a \emph{Bernoulli randomization design}, where each unit receives treatment independently with probability $p$. This Bernoulli design can be viewed as a special case of cluster-based randomization where each cluster consists of a single unit, i.e., $m=n$ and $|C_1| = |C_2| = \dots = |C_n| = 1$. By Lemma~\ref{lem: cluster-based expectation}, the mean of the difference-in-means estimator for the Bernoulli design is $\hat\tau_{b} := \bar\beta$. If we could obtain estimates from both $\hat\tau_{cb}$ and $\hat\tau_b$, we could construct a new estimator $\hat\tau' := 2\hat\tau_{cb} - \hat\tau_{b}$, which would be an unbiased estimator for the TTE with $\E[\hat\tau'] = \bar\beta + \gamma$. However, since each unit can only receive one treatment at a time, this approach is infeasible.

{This limitation motivates the introduction of a mixed randomization design, which combines cluster-based and Bernoulli designs. In the mixed design, a random decision is made for each unit to determine whether it will receive treatment according to the cluster-based or Bernoulli design, allowing us to simultaneously obtain estimates from both designs. In the first stage, the units are randomly partitioned into two subsets. In the second stage, one subset receives treatment based on the cluster-based design, while the other subset follows the Bernoulli design. Importantly, the randomization decisions in the two stages are independent.}

While similar mixed designs have been proposed by \cite{saveski2017detecting} and \cite{pouget2019hierachicaltesting}, their focus was on hypothesis testing for the presence of network interference. In contrast, our work aims to use the mixed design to estimate the TTE. The formal definition of the mixed design is provided below:

\begin{definition}[Mixed Randomization Design]{\label{def: mixed randomization design} The mixed randomization design is generated as follows: }
\begin{enumerate} 
    \item Let ${C_1, C_2, \dots, C_m }$ be a set of clusters forming a partition of the network $G(V, E)$. Let $\BFW \in {0,1}^m$ be a random vector indicating the assignment of each cluster to either the cluster-based or Bernoulli design. Specifically, $W_j = 1$ implies that cluster $C_j$ follows the cluster-based design, while $W_j = 0$ implies that cluster $C_j$ follows the Bernoulli design. The variables $W_j$ are i.i.d.\ Bernoulli random variables with mean $q$ ($\forall j \in [m]$). Let $\tilde{\BFw} \in {0,1}^n$ represent the corresponding unit-level assignment, where $\tilde{w}_i = W_{c(i)}$ for all $i \in V$.
    
    \item For cluster $C_j$, $j\in [m]$: if $W_j = 1$, assign all the units in this cluster treatment $z_i = 1$ with probability $p$, and assign all the units treatment $z_i = 0$ with probability $1-p$. If $W_j = 0$, assign treatment to the units in cluster $C_j$ by i.i.d.\ Bernoulli variables with mean $p$.
\end{enumerate}
\end{definition}

Note that this procedure does not specify how the clusters ${C_1, C_2, \ldots, C_m}$ should be formed. The design of clustering algorithms to optimize the efficiency of the mixed randomization design will be discussed in later section.

We define the following estimators: $\hat\tau_c$ for the units following the cluster-based design and $\hat\tau_b$ for the units following the Bernoulli design: \begin{subequations}
\begin{align}
    \hat\tau_{c}&=\frac{1}{nq}\sum_{i=1}^n \tilde{w}_{i}\left(\frac{z_i}{p}-\frac{1-z_i}{1-p}\right)Y_i(\BFz),\label{eq: cluster-based arm estimator}\\
    \hat\tau_{b}&=\frac{1}{n(1-q)}\sum_{i=1}^n (1-\tilde{w}_{i})\left(\frac{z_i}{p}-\frac{1-z_i}{1-p}\right)Y_i(\BFz).\label{eq: Bernoulli arm estimator}
\end{align}
\end{subequations}

{Under the mixed randomization design, we construct a novel estimator, called the mixed estimator, which is an affine combination of $\hat\tau_c$ and $\hat\tau_b$:}
{
\begin{definition}[The Mixed Estimator]\label{def: the mixed estimator}
    Under the mixed randomization design in Definition \ref{def: mixed randomization design}, we call $\hat\tau_m$ the mixed estimator if $\hat\tau_m :=\rho\hat\tau_c+(1-\rho)\hat\tau_b$, where $\hat\tau_c$ and $\hat\tau_b$ are defined in \eqref{eq: cluster-based arm estimator} and \eqref{eq: Bernoulli arm estimator}, respectively, and $\rho$ is a value specified before the experiment begins.
\end{definition}}

We will investigate the theoretical properties of the mixed estimator under known network effects in the following sections. The case where network effects are unknown will be discussed in Section~\ref{sec: unknown interfernece}.

{\section{The Mixed Estimator Under Known Interference}}\label{section: known interference}

{In this section, we consider the case where the interference, or network effect, is known to the experimenter in the post-experimental phase. Specifically, we assume access to the values of $v_{ij}$ for all $(i,j)\in E$, while other parameters in our model \eqref{eq:interference model} remain unknown. As an illustrative example, consider social network experimentation where interference arises from information sharing between individuals. Suppose the potential outcome is the amount of time an individual spends watching videos on a platform, which includes both platform-recommended videos and those shared by friends. Using historical data, experimenters can estimate the sharing frequency $v_{ij}$ for each pair of contacts $(i,j)$. However, the magnitude of this effect remains uncertain, so we introduce a parameter $\gamma$ in model \eqref{eq:interference model}, which captures the strength of the effect. Similar settings have been studied in various works \citep{brennan2022clusterbipartite, harshaw2023bipartitelinear, viviano2023causal-clustering}.}

{\subsection{Analysis of Bias and Variance}}

{With known weights $v_{ij}$ for all $(i,j)\in E$, it is possible to choose an appropriate value for the parameter $\rho$ such that the mixed estimator in Definition~\ref{def: the mixed estimator} is unbiased. This result is summarized in the following lemma:}

{\begin{lemma}\label{lemma: unbiased rho}
Under Assumption \ref{ass: linear model} and the mixed randomization design, the estimator in Definition \ref{def: the mixed estimator} is unbiased if and only if
     \begin{align}\label{eq: original rho}
        \rho := \frac{\sum_{i=1}^n \sum_{j\in \mathcal{N}_i} v_{ij}}{\sum_{i=1}^n \sum_{j\in \mathcal{N}_i} v_{ij} \boldsymbol{1}\{c(i)=c(j)\}}.
    \end{align}
\end{lemma}}

\begin{proof}
Substituting model \eqref{eq:interference model} into equation \eqref{eq: cluster-based arm estimator} and taking expectations, we get 
\[
\E[\hat\tau_c] = \bar\beta + \frac{\gamma}{n} \sum_{i=1}^n \sum_{j\in \mathcal{N}_i} v_{ij} \boldsymbol{1}\{c(i)= c(j)\}.
\]
Similarly, $\E[\hat\tau_b] = \bar\beta$. Using the definition of the mixed estimator $\hat\tau_m = \rho \hat\tau_c + (1-\rho) \hat\tau_b$, we have
\[
\E[\hat\tau_m] = \bar\beta + \rho \frac{\gamma}{n} \sum_{i=1}^n \sum_{j\in \mathcal{N}_i} v_{ij} \boldsymbol{1}\{c(i)= c(j)\},
\]
which equals the TTE if and only if Eq~\eqref{eq: original rho} holds.
\end{proof}

Computing the coefficient $\rho$ using Eq~\eqref{eq: original rho} requires full knowledge of the weight coefficients $v_{ij}$ for all edges $(i,j) \in E$. {The numerator of $\rho$ represents the total interference weights, while the denominator captures the total weights of the within-cluster connections. The variance of the mixed estimator depends on $\rho$, as well as on the variances of $\hat\tau_c$ and $\hat\tau_b$. To establish bounds on the variance of the mixed estimator, we first derive variance bounds for $\hat\tau_c$ and $\hat\tau_b$.}
{\begin{lemma}\label{lemma: cluster-based variance}
    Under Assumptions~\ref{ass: linear model} and \ref{ass: bounded outcomes}, and the mixed randomization design, the estimator defined in Eq~\eqref{eq: cluster-based arm estimator} has the following worst-case variance bound:
    \[
    \Var(\hat\tau_c) = \Theta\left(\frac{\eta}{qp(1-p)} + \delta \right),
    \]
where
\begin{align*}
\eta := \frac{1}{n^2} \sum_{k=1}^m |C_k|^2,\quad 
\delta := \frac{1}{n^2} \sum_{1\leq k \ne l \leq m} \left(\sum_{i\in C_k} \sum_{i' \in \mathcal{N}_i \cap C_l} v_{ii'} \right) \left(\sum_{j \in C_l} \sum_{j' \in \mathcal{N}_j \cap C_k} v_{jj'} \right).
\end{align*} 
\end{lemma}}

The proof can be found in Appendix~\ref{apx: proof of lemma cluster-based variance}. In particular, the bound depends on the number and sizes of the clusters ($\eta$), and the weights of between-cluster connections ($\delta$). {Larger clusters increase $\eta$, leading to higher variance for $\hat\tau_c$. The clustering quality also affects the variance through $\delta$, which is zero if the network is perfectly partitioned. Since the Bernoulli design is a special case of cluster-based design where each unit forms its own cluster, we can derive the variance bound for the estimator in Eq~\eqref{eq: Bernoulli arm estimator} by applying Lemma~\ref{lemma: cluster-based variance}:}
{
\begin{corollary}\label{corollary: bernoulli design estimator variance}
    Under Assumptions~\ref{ass: linear model} and \ref{ass: bounded outcomes}, and the mixed randomization design, the estimator defined in Eq~\eqref{eq: Bernoulli arm estimator} has the following worst-case variance bound:
\[
\Var(\hat\tau_b) = \Theta\left(\frac{1}{n(1-q)p(1-p)}\right).
\]
\end{corollary}}

\begin{proof}
Consider the definitions of $\eta$ and $\delta$ in Lemma~\ref{lemma: cluster-based variance}. Under the Bernoulli design, $\eta = \frac{1}{n^2} \sum_{k \in [n]} 1 = \frac{1}{n}$ and 
\[
|\delta| = \frac{1}{n^2} \left| \sum_{i=1}^n \sum_{\{j\in \mathcal{N}_i : i\in \mathcal{N}_j\}} v_{ij} v_{ji} \right| \leq \frac{1}{n},
\]
where the last inequality follows from Assumption~\ref{ass: linear model}, which states that $\sum_{j \in \mathcal{N}_i} |v_{ij}| \leq 1$ for all $i \in V$. The result follows by applying Lemma~\ref{lemma: cluster-based variance}.
\end{proof}

{With the variance bounds for $\hat\tau_c$ and $\hat\tau_b$, we can now derive a variance bound for the mixed estimator $\hat\tau_m = \rho \hat\tau_c + (1-\rho) \hat\tau_b$. Specifically, we have
\[
\Var(\hat\tau_m) = \rho^2 \Var(\hat\tau_c) + (1-\rho)^2 \Var(\hat\tau_b) + 2\rho(1-\rho) \Cov(\hat\tau_c,\hat\tau_b).
\]
By the Cauchy–Schwarz inequality, $|\Cov(\hat\tau_c,\hat\tau_b)| \leq \sqrt{\Var(\hat\tau_c) \Var(\hat\tau_b)}$, so
\begin{align}\label{eq: general variance bound}
\Var(\hat\tau_m) \leq (|\rho|\sqrt{\Var(\hat\tau_c)} + |1-\rho|\sqrt{\Var(\hat\tau_b)})^2 \leq 2\rho^2 \Var(\hat\tau_c) + 2(1-\rho)^2 \Var(\hat\tau_b).
\end{align}
This inequality, along with Lemma~\ref{lemma: cluster-based variance} and Corollary~\ref{corollary: bernoulli design estimator variance}, will be fundamental to our subsequent analysis. Using Eq~\eqref{eq: general variance bound}, we propose an efficient clustering algorithm to upper bound the variance of $\hat\tau_m$, then provide a lower bound for any clustering algorithm in Section~\ref{sec: Optimize the Experimental Design}.}

{\subsection{Optimize the Experimental Design}}\label{sec: Optimize the Experimental Design}
In this section, we discuss how to reduce the variance of the mixed estimator through proper clustering design. However, finding an exact algorithm to minimize the expression in Eq~\eqref{eq: general variance bound} is computationally challenging, as it requires jointly optimizing $\rho$, $\eta$, and $\delta$. \cite{andreev2004balanced} showed that no polynomial-time approximation algorithm can guarantee a finite approximation ratio in graph partitioning (in our case, a finite bound on $\rho$) unless $\mathcal{P} = \mathcal{NP}$. Therefore, instead of searching for an optimal algorithm to minimize the bound in Eq~\eqref{eq: general variance bound}, 
we resort to finding an efficient clustering algorithm that leads to an asymptotic bound on $\Var(\hat\tau_m)$. 

We first consider arbitrary network structures and 
propose a clustering algorithm based on a greedy heuristic. The key idea is to use maximum weight matching to generate clusters of sizes 1 and 2, and then merge these small clusters together to create the desired clustering. The merging step follows a greedy rule, where we iteratively look for two clusters to merge, minimizing the right-hand side of Eq~\eqref{eq: general variance bound}. The objective function for this greedy algorithm is written as:
{
\begin{align}\label{eq: greedy objective}
    S(\rho,\eta,\delta) = \frac{\rho^2\eta}{q} + \rho^2\delta p(1-p) + \frac{(1-\rho)^2}{n(1-q)}.
\end{align}
}
{
The function $S$ is obtained from Eq~\eqref{eq: general variance bound} and the bounds in Lemma~\ref{lemma: cluster-based variance} and Corollary~\ref{corollary: bernoulli design estimator variance}. The values of $\rho$ and $\eta$ can be calculated incrementally when two clusters are merged, while the computation of $\delta$ is more challenging. We instead calculate an upper bound for $\delta$, which depends on the maximum size of the clusters:
\begin{lemma}\label{lem: size bound}
We have the following bound for $\delta$ (defined in Lemma~\ref{lemma: cluster-based variance}): 
\[
\delta \le \tilde{\delta} = \frac{\max_{k\in[m]} |C_k|}{n}.
\]
\end{lemma}
}
{The proof for this lemma can be found in Appendix~\ref{apx: proof of lemma size bound}. Therefore, we replace $\delta$ with its upper bound $\tilde{\delta}$ in Eq~\eqref{eq: greedy objective}. With a slight abuse of notation, let $S(\{C_1,...,C_m\}) \equiv S(\rho,\eta,\tilde{\delta})$. Starting from an initial solution generated by maximum weight matching, we iteratively find pairs of clusters to merge that maximize the change in the objective function. Formally, suppose the current solution is $\BFC = \{C_1,...,C_m\}$.} After merging two clusters $C_k, C_l \in \BFC$, we obtain a new cluster set $\BFC' = \BFC / \{C_k, C_l\} \cup \{C_k \cup C_l\}$, and the corresponding objective function decreases by:
\begin{align}\label{eq: difference of objective function}
    \Delta S_{\BFC}(C_k,C_l) = S(\BFC) - S(\BFC / \{C_k, C_l\} \cup \{C_k \cup C_l\}).
\end{align}
The stopping criterion is to check whether $\Delta S_{\BFC}(C_k, C_l)$ is positive, indicating whether the approximate variance upper bound can be further improved by merging clusters. We present the greedy heuristic in Algorithm~\ref{alg: greedy algorithm}.

\begin{algorithm}[!tbp]
\label{alg: greedy algorithm}
\DontPrintSemicolon
  \KwInput{A graph $G(V,E)$}
  \KwOutput{A clustering $\{C_1,C_2,\ldots\}$ of the graph}

    Find a maximum weight matching $M = \{(i_1,j_1), (i_2,j_2), ..., (i_m,j_m)\}$ of the graph $G$ using the weights $v_{ij}$ in Eq~\eqref{eq:interference model}\;
    Initialize the clusters: $\BFC \leftarrow M$\;
    \For{any $k \in [n]$ such that $\nexists j \in \mathcal{N}_k$ with $(k,j) \in M$ or $(j,k) \in M$}     
    { 
        $\BFC \leftarrow \BFC \cup \{\{k\}\}$\;
    }
    \While{$\operatornamewithlimits{max}\limits_{C_{k_1}, C_{k_2} \in \BFC, k_1 \ne k_2} {\Delta S_{\BFC}(C_{k_1}, C_{k_2})} > 0$}{
       $(C_k, C_l) \leftarrow \operatornamewithlimits{argmax}\limits_{C_{k_1}, C_{k_2} \in \BFC, k_1 \ne k_2} {\Delta S_{\BFC}(C_{k_1}, C_{k_2})}$\Comment*[r]{finding two clusters to merge}
       $\BFC \leftarrow \BFC / \{C_k, C_l\} \cup \{C_k \cup C_l\}$\;
    }
    \Return $\BFC$\;
\caption{Greedy Clustering}
\end{algorithm}

The computational complexity of Algorithm~\ref{alg: greedy algorithm} is $O(n^3)$. First, finding a maximum matching in a general graph with $n$ vertices is $O(n^3)$ \citep{galil1986maximummatching}, so lines 1 to 4 in Algorithm~\ref{alg: greedy algorithm} take $O(n^3)$ time. Second, the while loop from line 5 to line 7 can iterate at most $n$ times since there are $n$ vertices. Finally, line 6 can be implemented in $O(|E|)$ time by enumerating the edge set $E$ to find all adjacent clusters and calculating $\sum_{i \in C_k} \sum_{i' \in \mathcal{N}_i \cap C_l} v_{ii'}$ for every adjacent pair $(C_k, C_l)$. Thus, lines 5 to 7 in Algorithm~\ref{alg: greedy algorithm} take $O(n|E|)$ time, where $E$ is the edge set and $|E| \le n(n-1)$.

Throughout the paper, we use $d := \max_{i \in V} |\mathcal{N}_i|$ to denote the maximum degree of the network. {We also assume the design parameter $q$, the probability that a cluster is assigned to the cluster-based design, is a fixed constant.}
The variance of $\hat\tau_m$ using the output of Algorithm~\ref{alg: greedy algorithm} is bounded as follows:

{\begin{theorem}\label{theorem: greedy algorihtm upper bound} 
Under Assumptions~\ref{ass: linear model}--\ref{ass: bounded outcomes} and the mixed randomization design, the variance of the mixed estimator $\hat\tau_m$ using the clustering from Algorithm \ref{alg: greedy algorithm} is upper bounded by $O\left(\frac{d^2}{np(1-p)}\right)$.
\end{theorem}}

The complete proof is included in Appendix~\ref{apx: proof of theorem greedy algorihtm upper bound}.
Theorem~\ref{theorem: greedy algorihtm upper bound} guarantees an $O(d^2)$ upper bound on the variance of $\hat\tau_m$ for general networks.
For networks where $d = o(\sqrt{n})$, the estimator $\hat\tau_m$ is consistent, as the variance converges to 0 as the network size $n \to \infty$. To put these upper bounds in perspective, we present a lower bound on the variance of $\hat\tau_m$ for any clustering algorithm. To establish the lower bound, we consider a specific family of networks with cycle structures.

\begin{definition}[Cycle network]\label{def: cycle network}
We call a network $(d,\kappa)$-cycle with $1 \leq \kappa \leq d$ if each unit $i \in [n]$ is connected to the units indexed by $\{(i\pm 1)\mod n, (i\pm 2)\mod n, ..., (i\pm (\kappa-1))\mod n \}$ and $\{(i\pm \kappa)\mod n, (i\pm 2\kappa)\mod n, ..., (i\pm d\kappa)\mod n \}$.
\end{definition}

It is easily verified that a cycle network in Definition~\ref{def: cycle network} has a maximum degree of $2(d+\kappa) \leq 4d$. Therefore, the parameter $d$ controls the magnitude of the maximum degree, {while $\kappa$ controls the maximum distance a unit can reach through the connections.} The following theorem establishes a lower bound on $\Var(\hat\tau_m)$ for any clustering algorithm under the mixed randomization design.
{
\begin{theorem}\label{The: lower bound cycle graph}
Consider a $(d,\kappa)$-cycle network whose potential outcomes are given by model~\ref{eq:interference model} with parameters $\alpha_i = 1$, $\beta_i = 0$ for all $i \in [n]$, $\gamma = 0$, and $v_{ij} = |\mathcal{N}_i|^{-1}$ for all $(i,j) \in E$.
Then under any clustering algorithm, the estimator $\hat\tau_m$ in the mixed design (Definition~\ref{def: mixed randomization design}) satisfies 
\[
\Var(\hat\tau_m) = \Omega\left(\frac{\min\{\kappa d, d^2/\kappa\}}{np(1-p)}\right).
\]
\end{theorem}
}

The proof is given in Appendix~\ref{apx: proof of theorem lower bound cycle graph}. For general networks,
by setting $\kappa = \sqrt{d}$, the above theorem implies a $\Omega(d^{1.5}/(np(1-p)))$ lower bound for networks with a maximum degree $d$. This lower bound differs from the $O(d^2/(np(1-p)))$ upper bound in Theorem~\ref{theorem: greedy algorihtm upper bound} by only a factor of $O(\sqrt{d})$, verifying the performance of the proposed greedy algorithm. Notice that the lower bound in Theorem~\ref{The: lower bound cycle graph} holds even when using \emph{randomized} clusters, which will be studied in Section~\ref{sec: unknown interfernece}.

{
\subsection{Comparison with Regression Adjustment}\label{sec: compare with regression adjustment}}
{
This section is devoted to discussing why simpler approache, such as the regression adjustment estimator, cannot replace the mixed estimator in our setting. We demonstrate that the regression adjustment estimator is not unbiased under the model we study. The regression adjustment estimator is defined as follows:
}
{
\begin{definition}[Regression Adjustment Estimator]\label{def: regression adjustment}
    The regression adjustment estimator $\hat\tau_{rg}$ is constructed through the following process: Fit an ordinary least squares (OLS) regression model for $Y_i$, $i \in [n]$, over covariates $(z_i, \sum_{j \in \mathcal{N}_i} v_{ij} z_j)$, $i \in [n]$, and obtain the estimated linear model: 
    \[
    Y_i \approx \hat\alpha + \hat\beta z_i + \hat\gamma \sum_{j \in \mathcal{N}_i} v_{ij} z_j.
    \]
    Then use the coefficients to construct an estimator for the TTE, written as 
    \[
    \hat\tau_{rg} = \hat\beta + \hat\gamma \sum_{j \in \mathcal{N}_i} v_{ij}.
    \]
\end{definition}
}
{
\cite{chin2019regression} provided sufficient conditions for the regression adjustment estimator to be unbiased:
\begin{lemma}[\cite{chin2019regression}]\label{lemma: regression adjustment}
    The estimator defined in Definition \ref{def: regression adjustment} is unbiased if the following conditions hold:
    \begin{enumerate}
        \item The data-generation process is $Y_i = \alpha + \beta z_i + \gamma X_i(\BFz) + \epsilon_i$, where $X_i(\BFz) = \sum_{j \in \mathcal{N}_i} v_{ij} z_j$ is the covariate for the interference, and $\epsilon_i$ is the error term.
        \item $X_i(\BFz)$ and $z_i$ are independent for all $i \in [n]$.
        \item The errors are strictly exogenous: $E[\epsilon_i | X_1(\BFz), ..., X_n(\BFz)] = 0$ for all $i \in [n]$.
    \end{enumerate}
\end{lemma}
}

Our model \eqref{eq:interference model} satisfies the first condition of Lemma~\ref{lemma: regression adjustment} if we set $Y_i = \bar\alpha + \bar\beta z_i + \gamma X_i(\BFz) + \epsilon_i$, where $\epsilon_i = \alpha_i - \bar\alpha + (\beta_i - \bar\beta) z_i$ and $\bar\alpha$ and $\bar\beta$ are the means of $\alpha_i$ and $\beta_i$, respectively, for $i \in [n]$. Since the neighborhood set $\mathcal{N}_i$ does not include $i$ itself, the second condition in the lemma also holds. 

However, the third condition, which requires the errors to be strictly exogenous, does not hold in our setting. This is because $\epsilon_i$ is affected by the treatment assignment $z_i$, which, in turn, affects $X_j(\BFz)$ if $i \in \mathcal{N}_j$. The reason for this is that we allow for heterogeneous treatment effects ($\beta_i$), while the regression adjustment approach only accounts for homogeneous treatment effects. 

We provide a numerical counterexample in Appendix~\ref{apx: additional numerical result}, showing that the regression adjustment approach suffers from severe bias under certain circumstances, whereas the mixed estimator remains unbiased.

{\section{Experimental Design under Unknown Interference}}\label{sec: unknown interfernece}

The mixed experiment design proposed in Section~\ref{section: known interference} requires full knowledge of the weight coefficients $v_{ij}$ for all edges $(i,j) \in E$. However, in many practical applications, these weights are unknown a priori. In this section, we extend the mixed experiment design by constructing a clustering algorithm with a corresponding TTE estimator that is agnostic to the weight coefficients.

Recall that the mixed randomization design is based on clustering the network. In the previous section, we assumed the clusters were fixed. The key to extending the design to unknown weights is to use \emph{randomized} clusters. Our approach is motivated by \cite{ugander2023randomized}, who showed that randomized clustering algorithms can provide better variance upper bounds for cluster-based designs. 

Consider a randomized clustering algorithm that produces $k$ different clusterings $\BFC_1, \BFC_2, \dots, \BFC_k$ for a graph $G(V, E)$. Each clustering $\BFC_l$ ($\forall l \in [k]$) is a set of clusters that forms a partition of the network, i.e., $\BFC_l = \{C_1^l, C_2^l, \dots, C_{m_l}^l\}$. Suppose $\BFC_l$ is generated by the clustering algorithm with probability $p(l)$, where $\sum_{l=1}^k p(l) = 1$. Let $c_l(i): V \to [m_l]$ be the function that maps each unit $i \in V$ to the cluster index in the $l^{\text{th}}$ clustering. Then, we can generalize the proof of Lemma~\ref{lemma: unbiased rho} to show that:
\begin{align}\label{eq: randomized expectation of mixed estimator}
    \E[\hat\tau_m] = \bar\beta + \rho \frac{\gamma}{n} \sum_{i=1}^n \sum_{j \in \mathcal{N}_i} v_{ij} \sum_{l=1}^k p(l) \boldsymbol{1}\{c_l(i) = c_l(j)\} 
    = \bar\beta + \rho \frac{\gamma}{n} \sum_{i=1}^n \sum_{j \in \mathcal{N}_i} v_{ij} P(c(i) = c(j)).
\end{align}

{We say that a randomized clustering algorithm is \emph{weight-invariant} if the value of $\E[\hat\tau_m]$ equals the TTE for a specific value of $\rho$ that does not depend on the weights $v_{ij}$ for all $(i,j) \in E$. The only way to achieve this is if there exists $\rho \geq 1$ such that $P(c(i) = c(j)) = 1/\rho$ for all $(i,j) \in E$. We summarize this result in the following lemma.
\begin{lemma}\label{lemma: weight invariant design}
Under Assumption~\ref{ass: linear model} and the mixed randomization design, if there exists a randomized clustering algorithm and a parameter $\rho \geq 1$ such that $P(c(i) = c(j)) = 1/\rho$ for all $(i,j) \in E$, then the mixed estimator $\hat\tau_m$ (defined in Definition~\ref{def: the mixed estimator}) is unbiased.
\end{lemma}
}

The proof follows directly from Eq~\eqref{eq: randomized expectation of mixed estimator} and is thus omitted. The main challenge in TTE estimation under unknown interference lies in designing the weight-invariant clustering algorithm. We will describe the procedure next. {Our algorithm is related to edge partitioning, which divides $E$ into $u$ disjoint subsets, $E = E_1 \cup E_2 \cup \dots \cup E_u$. Let $V(E_i)$ denote the set of all nodes incident to edges in $E_i$. We say an edge partitioning is proper if every partition includes every edge whose vertices are incident to edges in it. Below we give a formal definition.
\begin{definition}[Proper Edge Partitioning]\label{def: proper edge partitioning}
 The edge partitioning $E = E_1 \cup E_2 \cup \dots \cup E_u$ is proper if:
 \begin{enumerate}
     \item $E_1, \dots, E_u$ are disjoint;
     \item $E_k = \{(i,j) \in E \mid i \in V(E_k), j \in V(E_k)\}$ for all $k \in [u]$.
 \end{enumerate}
\end{definition}
}
{
Note that a proper edge partitioning always exists for any network. For example, if each edge itself forms a partition, the conditions in Definition~\ref{def: proper edge partitioning} trivially hold. Another example is the partitioning generated by cliques. A clique is a complete sub-graph of $G$ and thus satisfies the conditions in Definition~\ref{def: proper edge partitioning}. We also mention the edge partitioning algorithm proposed by \cite{zhang2017neighborhood-expansion}, known as the \textit{Neighborhood Expansion} heuristic. This heuristic seeks to minimize the number of replicated nodes and always generates a proper edge partitioning. We will use this heuristic in our numerical experiments in Section~\ref{sec:simulation}.
}

{
After obtaining a proper edge partitioning, we use the weight-invariant clustering algorithm to transform the edge partitioning into a clustering. We define the incidence matrix of an edge partitioning $E = E_1 \cup E_2 \cup \dots \cup E_u$ as a $u \times u$ binary symmetric matrix $M$, where $M_{ij} = 1$ if $V(E_i) \cap V(E_j) \neq \emptyset$ and $M_{ij} = 0$ otherwise for all $i,j \in [u]$. This procedure is presented in Algorithm~\ref{alg: weight-invariant clustering}.
}

\begin{algorithm}[!thbp]\label{alg: weight-invariant clustering}
\DontPrintSemicolon
\KwInput{A graph $G(V,E)$}
    Initialize the clustering $\BFC\leftarrow\{\}$, {the edge partitioning $E=E_1\cup E_2\cup ...\cup E_u$}, and the incidence matrix $M\in\{0,1\}_{u\times u}$\;
    Find the maximum eigenvalue $\lambda^*$ of $M$ and a corresponding eigenvector $\boldsymbol{\omega}$\;
    \For{$i\in [u]$}{
        $X_{i}\leftarrow \text{Exp}(\omega_{i})$ \Comment*[r]{
        sample from Exponential distribution}
    }
    \For{$i\in [u]$}{
        \If{$X_{i}<X_{j}$ for all $j\ne i: M_{ij}=1$}{
            $\BFC\leftarrow \BFC\cup\{V(E_{i})\}$ \Comment*[r]{
        generate clusters from edge partitions}
    }
    }
    \For{$v\in V$}{
        \If{$v\notin C_k\;\forall k\in [|\BFC|]$}{
        $\BFC\leftarrow \BFC\cup\{\{v\}\}$\;
        }
    }
  \KwOutput{Clustering $\BFC$}
\caption{Weight-invariant Clustering}
\end{algorithm}

{In Algorithm~\ref{alg: weight-invariant clustering}, we first calculate the eigenvector and maximum eigenvalue of the incidence matrix. Then, we assign a weight to each edge partition, which is sampled from an exponential distribution based on the eigenvector. Next, we check each edge partition to see if the assigned weight is smaller than all weights from the incident edge partitions. If true, we assign all vertices of that edge partition into a cluster. Finally, if there are any units that have not been assigned to a cluster yet, we let each unassigned unit form its own cluster. Note that each unit is assigned to exactly one cluster in the output of Algorithm~\ref{alg: weight-invariant clustering}, as line 6 ensures that two incident edge partitions will not be used to generate clusters simultaneously. Therefore, Algorithm~\ref{alg: weight-invariant clustering} is a valid clustering algorithm.}

The next theorem shows that Algorithm~\ref{alg: weight-invariant clustering} is weight-invariant if we choose $\rho$ to be the maximum eigenvalue $\lambda^*$ of the incidence matrix $M$.
\begin{theorem}\label{theorem: the algorithm is weight invariant}
    Under Assumption~\ref{ass: linear model} and the mixed randomization design, if the clustering is generated by Algorithm~\ref{alg: weight-invariant clustering} with a proper edge partitioning, and the parameter of the mixed estimator is set as the maximum eigenvalue $\rho = \lambda^*$, then the mixed estimator $\hat\tau_m$ is unbiased.
\end{theorem}

The proof can be found in Appendix~\ref{apx: proof of theorem algorithm is weight invariant}. We briefly describe the key ideas below. Due to the property of a proper edge partitioning, for $(i,j) \in E_k$, the probability that $i$ and $j$ are in the same cluster equals the probability that $V(E_i)$ is selected as a cluster in line 7 of Algorithm~\ref{alg: weight-invariant clustering}. This probability is equivalent to the probability that $X_i < X_j$ for all $j \neq i$ such that $M_{ij} = 1$. The value of this probability is $\omega_i / \sum_{j: M_{ij} = 1} \omega_j$ since the $X_i$'s are generated from the exponential distribution. 

Since $\lambda^*$ is the eigenvalue, we have $M\omega = \lambda^* \omega$, which implies that $\omega_i / \sum_{j: M_{ij} = 1} \omega_j = (\lambda^*)^{-1}$ for all $i \in [u]$. Using this result, we can apply Lemma~\ref{lemma: weight invariant design} to show the unbiasedness of $\hat\tau_m$.

Finally, we give an upper bound on the variance of $\hat\tau_m$ under Algorithm \ref{alg: weight-invariant clustering}. This upper bound differs from the bound in Theorem~\ref{theorem: greedy algorihtm upper bound} by a factor of $O(d)$, since Algorithm \ref{alg: weight-invariant clustering} does not require the knowledge of the weight coefficients $v_{ij}$, $\forall (i,j) \in E$.

\begin{theorem}\label{theorem: weight invariant variance}
   Under Assumption~\ref{ass: linear model}$\sim$\ref{ass: bounded outcomes} and the mixed randomization design, if we use Algorithm \ref{alg: weight-invariant clustering} with edge partitioning $E=E_1\cup...\cup E_u$ satisfying $|V(E_i)|=O(1)$, $\forall i\in [u]$, then the variance of the unbiased mixed estimator ($\rho=\lambda^*$) has the following upper bound
   \[\Var(\hat\tau_m)=O\left(\frac{d^3}{np(1-p)}\right)\]
\end{theorem}
The proof can be found in Appendix~\ref{apx: proof of theorem weight invariant variance}. The key idea of proof is to use $\Var(\hat\tau_m)=\Var(\E[\hat\tau_m|\BFC])+\E[\Var(\hat\tau_m|\BFC)]$ and bound the two terms in the right-hand side separately. The bound for the first term relies on a delicate analysis of the edge partitioning, while the bound for the second term follows the procedure similar to the proof of Theorem~\ref{theorem: greedy algorihtm upper bound}.

\section{Inference}\label{section: Inference}

In this section, we provide useful results for performing statistic inference on the TTE, e.g., by hypothesis testing or constructing confidence intervals. Firstly, we provide exact upper bound for the variance of the mixed estimator under Algorithm~\ref{alg: greedy algorithm} and \ref{alg: weight-invariant clustering}, which serve as the conservative variance estimators for the mixed estimator under both conditions. 
\begin{lemma}\label{lem: exact upper bound greedy}
    Under Assumption~\ref{ass: linear model}$\sim$\ref{ass: bounded outcomes} and Algorithm~\ref{alg: greedy algorithm}, we have
    \[\Var(\hat\tau_m)\le \widehat\Var_1(\hat\tau_m) = \left(|\rho|\sqrt{\widehat\Var(\hat\tau_c)}+|\rho-1|\sqrt{\widehat\Var(\hat\tau_b)}\right)^2,\]
    where 
    \[\widehat\Var(\hat\tau_c)= \sqrt{\left(\frac{Y_M^2}{qp(1-p)}+2Y_M(Y_M-Y_L)\right)\eta+\delta},\]
    \[\widehat\Var(\hat\tau_b)= \sqrt{\left(\frac{Y_M^2}{(1-q)p(1-p)}+2Y_M(Y_M-Y_L+1)\right)/n},\]
    and $\rho$ is defined in Lemma~\ref{lemma: unbiased rho}, $\eta$ and $\delta$ are defined in Lemma~\ref{lemma: cluster-based variance}
\end{lemma}

The proof is based on Eq~\eqref{eq: general variance bound}. The exact variance upper bound for $\hat\tau_c$ and $\hat\tau_b$ (i.e. $\widehat\Var(\hat\tau_c)$ and $\widehat\Var(\hat\tau_b)$) are from the proof of Lemma~\ref{lemma: cluster-based variance} and therefore the proof for the above Lemma is omitted.

\begin{lemma}\label{lem: exact upper bound weight invariant}
    Under Assumption~\ref{ass: linear model}$\sim$\ref{ass: bounded outcomes} and Algorithm~\ref{alg: weight-invariant clustering}, we have
    \[\Var(\hat\tau_m)\le \widehat\Var_2(\hat\tau_m)= \bar E + \frac{\bar V^2}{n} (Y_M-Y_L)^2d^2\lambda^*,\]
    where $\bar V=\max_{i}|V(E_i)|$, $\lambda^*$ is the largest eigenvalue in Algorithm~\ref{alg: weight-invariant clustering} and $\bar E$ equals to the value of $\widehat\Var_1(\hat\tau_m)$ in Lemma~\ref{lem: exact upper bound greedy} with $\rho=\lambda^*$, $\eta=\delta=\bar V/n$.
\end{lemma}

Lemma~\ref{lem: exact upper bound weight invariant} is based on the proof of Theorem~\ref{theorem: weight invariant variance}. We omit the proof here since it is basically retrieving a constant to convert the "big O" notation to an exact upper bound. With these variance estimators, we are able to test whether the treatment leads to a change in the TTE. With the null hypothesis being that the TTE is zero, and the estimated variance being $\hat\sigma^2$, we can construct a confidence interval based on the Chebyshev's inequality, which states
\begin{align*}
    P(|\hat\tau_m-E(\hat\tau_m)|>t\hat\sigma)\le \frac{1}{t^2},
\end{align*}
for any real number $t>0$. By rejecting the null hypothesis when $|\hat\tau_m|>\hat\sigma/\sqrt{\alpha}$, the type-I error of our test is guaranteed to be no greater than $\alpha$. However, this approach may be too conservative in practice, leading to less statistical power. A less conservative approach assumes that $(\hat\tau_m-E(\hat\tau_m))/\hat\sigma$ follows a standard normal distribution, allowing us to construct the $(1-\alpha)\times 100\%$ confidence interval
\[(\hat\tau_m+ \hat\sigma z_{\alpha/2}, \hat\tau_m+ \hat\sigma z_{1-\alpha/2}),\]
where $z_{\alpha/2}$ and $z_{1-\alpha/2}$ are the $\alpha/2$ and $1-\alpha/2$ quantiles of the standard normal distribution.
Typically, proofs for asymptotic normality rely on certain versions of the Central Limit Theorem. To this end, we rewrite the estimator $\hat\tau_m$ in the following form:
\begin{align}\label{eq: estimator reformulate for clt}
    \hat\tau_m&=\frac{1}{n}\sum_{i=1}^n 2(2\rho \tilde{w}_{i}-\rho-\tilde{w}_{i}+1)\left(\frac{z_i}{p}-\frac{1-z_i}{1-p}\right)Y_i(\BFz),
\end{align}
where $\tilde{w}_{i}$ indicates whether unit $i$ received a cluster-based treatment assignment or a Bernoulli assignment (see Eq~\eqref{eq: Bernoulli arm estimator} and \eqref{eq: cluster-based arm estimator}).
Let $L_i$ be the term within the above summation. Then $\hat\tau_m=\frac{1}{n}\sum_{i=1}^n L_i$ is the average of $n$ dependent random variables. We say the random variable $L_i$ has dependent neighborhood $N_i\subset [n]$, if $i\in N_i$ and $X_i$ is independent of $\{L_j\}_{j\ne N_i}$. The following result is an extension of Stein's method for bounding the distance between probability distributions in the Wasserstein metric.

\begin{lemma}[\citealp{ross2011Stein}, Theorem 3.6]\label{lem: wasserstein bound}
    Let $X_1, \ldots , X_n$ be random variables such that $\E[X_i^4] < \infty$, $\E[X_i] =0$,
$\sigma^2 = \Var(\sum_{i=1}^n X_i)$, and define $W =
\sum_{i=1}^n X_i/\sigma$. Let $D=|\max_i\{j| X_i \text{ and } X_j \text{ are dependent}\}|$. Then for a standard normal random variable $Z$, we have
\begin{align*}
    d_W(W,Z)\le \frac{D^2}{\sigma^3}\sum_{i=1}^n \E[|X_i|^3]+ \frac{\sqrt{28} D^{3/2}}{\sqrt{\pi}\sigma^2}\sqrt{\sum_{i=1}^n \E[|X_i|^4]},
\end{align*}
where $d_W(\cdot,\cdot)$ is the Wasserstein metric:
\begin{align*}
    d_W(\mu,\upsilon)\ :=\ \sup_{\{h\in \mathbb{R}\rightarrow\mathbb{R}:|h(x)-h(y)|\leq |x-y|\}}\left|\int h(x)d \mu(x)-\int h(x)d \upsilon(x)\right|.
\end{align*}
\end{lemma}

By Lemma~\ref{lem: wasserstein bound}, we establish sufficient conditions for asymptotic normality of the estimator $\hat\tau_m$.

\begin{theorem}
\label{thm:asymptotic_normality}
    Define $\sigma^2_n=\Var({\sqrt{n}}\hat\tau_m/{\rho})$. Under Assumption~\ref{ass: linear model}$\sim$\ref{ass: bounded outcomes} and the mixed randomization design, if $\liminf_{n\rightarrow\infty} \sigma^2_n>0$, the treatment probability $p$ is fixed, and the maximum cluster size is $O(1)$, we have
    \begin{align*}
        \frac{\sqrt{n}(\hat\tau_m-\E[\hat\tau_m])}{\sigma_n \rho}\mathop{\rightarrow}\limits^{d} N(0,1)
    \end{align*}
    under one of the following additional conditions:

    (1) The edge weights are know, the design applies Algorithm~\ref{alg: greedy algorithm}, and ${d^8}/{n}\rightarrow 0$;

    (2) The edge weights are unknow, the design applies Algorithm~\ref{alg: weight-invariant clustering}, and ${d^{12}}/{n}\rightarrow 0$.
\end{theorem}

The proof can be found in Appendix~\ref{apx: proof of asymptotic_normality}. Therefore, as long as the graph is sufficiently sparse (e.g. the maximum degree is a constant), we can approximate the distribution of the mixed estimator by normal distribution under large population.

\section{Numerical Experiments}
\label{sec:simulation}

\subsection{Test Instances} \label{section: test instances}
We generalize networks in the numerical experiments from the random geometric graph (RGG) model \citep{ugander2013graph}. The RGG model is a spatial graph where $n$ units are scattered according to a uniform distribution in the region $[0,\sqrt{n}]\times[0,\sqrt{n}]$. Two units $i$ and $j$ are connected if $dist(i,j)\le \sqrt{r_0/\pi}$, where $r_0$ is the limiting expected degree of the RGG model. To simulate the long-range dependency, we allow a unit to connect with $r_1$ units outside the radius $\sqrt{r_0/\pi}$, which are selected uniformly at random. Figure~\ref{fig: test instances} shows two randomly generated RGG model with different parameters.
Furthermore, if $i$ and $j$ are connected, we assign the weight coefficient $v_{ij}$ independently from a uniform distribution $U(-1/r,2/r)$, where $r=r_0+r_1$. Note that the magnitude of the inference is inversely proportional to the average degree.
\begin{figure}[!htb]
\caption{Examples of RGG networks: the left figure is randomly generated from $(n,r_0,r_1)=(100,10,0)$ and the right figure $(n,r_0,r_1)=(100,5,5)$.}\label{fig: test instances}
\centering
\includegraphics[width=0.395\textwidth]{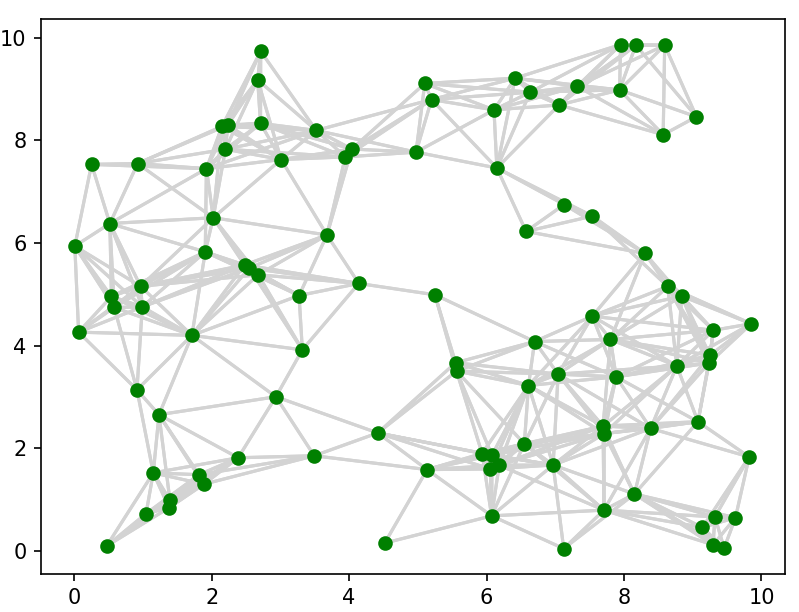}
\includegraphics[width=0.4\textwidth]{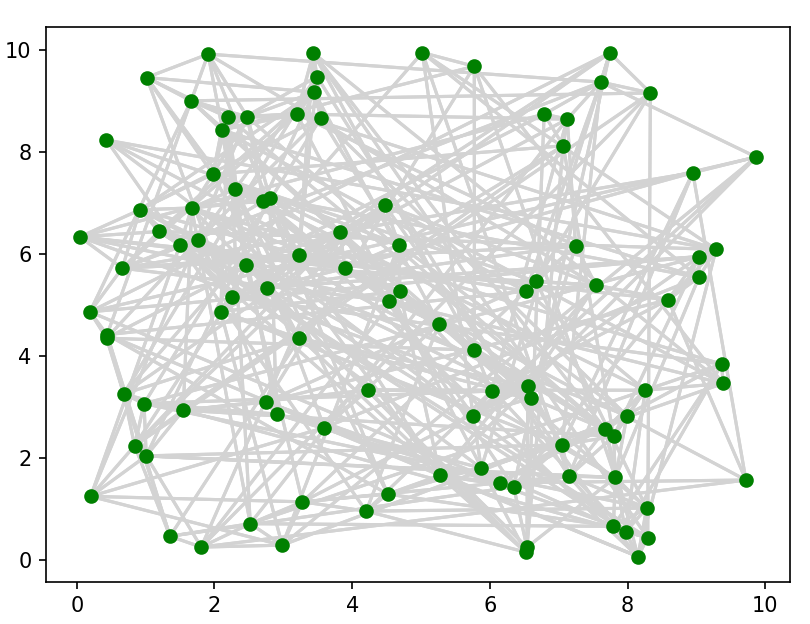}
\end{figure}

To choose the parameters in the exposure mapping model Eq~\eqref{eq:interference model}, 
we first generate $\hat{\alpha}_i$ and $\hat{\beta}_i$ independently from $U(-1,1)$ for each $i\in [n]$. We then re-scale  the parameters by $\alpha_i=\hat\alpha_i+1-\frac{1}{n}\sum_{j=1}^n \hat\alpha_j$, $\beta_i=\hat\beta_i+0.5-\frac{1}{n}\sum_{j=1}^n \hat\beta_j$ and $\gamma=0.5\left(\sum_{i=1^n}\sum_{j\in \mathcal{N}_i}v_{ij} \right)^{-1}$. 
The re-scaling step ensures that the true baseline parameter,
, the direct treatment effect parameter, and the interference effect parameter are $\bar\alpha=1$, $\bar\beta=0.5$, and $\gamma\sum_{i=1^n}\sum_{j\in \mathcal{N}_i}v_{ij}=0.5$, respectively.
The true value of the TTE is equal to $1$ in all randomly generated instances. We also fix the treatment probability $p=0.5$ and the probability of cluster-based randomization under the mixed design to be $q=0.9$, as the Bernoulli randomization only contributes a small proportion of variance under the mixed design.

{\subsection{Compare the Mixed Design with the Cluster-based Design}}\label{section: numerical comparison known interference}

{The first task of our numerical experiment is to verify that the mixed design indeed improves upon the cluster-based design by eliminating the bias from imperfect network clustering. In this experiment, we use test instances with known edge weights, which are required by clustering algorithms in practice. The benchmark is the cluster-based design (Definition~\ref{def:cbr}) generated by the state-of-the-art community detection algorithm called \textit{Leiden} \citep{leiden}, which maximizes the modularity of the clustering. We test the mixed design (Definition~\ref{def: mixed randomization design}) in conjunction with Algorithm~\ref{alg: greedy algorithm}. }

To compare the performance under these two experimental designs, we generate $3 \times 6 = 18$ test instances, parameterized by $(n, r_0, r_1)$. The network size $n$ is chosen from $\{1000, 2000, 4000\}$. The values of $(r_0, r_1)$ are chosen from the set $\{(4,0), (2,2), (0,4), (16,0), (8,8), (0,16)\}$. Note that $r_0 + r_1$ is either 4 or 16, representing the expected average degree in the random geometric graph (RGG) model. Recall that when $r_0$ is positive and $r_1$ is zero, units are only connected to those within short distances. When $r_1$ is positive, units may be subject to long-range interference.

{We consider three performance metrics for the TTE estimators $\hat\tau_m$ from the mixed design and $\hat\tau_{cb}$ from the cluster-based design: empirical bias (Bias), empirical variance ($\Var$), mean squared error (MSE), and type-1 error. The empirical bias and variance are calculated from 10,000 independent replications for each design. To compute the type-1 error, we first construct a $95\%$ confidence interval for TTE using the normal distribution, where the variance of $\hat\tau_m$ is estimated using Lemma~\ref{lem: exact upper bound greedy}. The estimated variance of $\hat\tau_{cb}$ is based on the well-known Neyman conservative variance estimator, which assumes that the clusters are independent of each other. Then, the type-1 error is obtained from 1000 simulations. The results are summarized in Table~\ref{table: compare with CBR}.}

\begin{table}[!htbp]
\centering
\begin{tabular}{l|rr|rr|rr| rr} 
 \hline\hline
 \multirow{2}{*}{$(n,r_0,r_1)$} & \multicolumn{2}{c|}{Bias} & \multicolumn{2}{c|}{Var} & \multicolumn{2}{c|}{MSE}& \multicolumn{2}{c}{Type-1 Error} \\
 \cline{2-9}
  & $\hat\tau_m$ & $\hat\tau_{cb}$ & $\hat\tau_m$ & $\hat\tau_{cb}$ & $\hat\tau_m$ & $\hat\tau_{cb}$ & $\hat\tau_m$ & $\hat\tau_{cb}$ \\  
 \hline
 (1000,4,0)& \bf 0.0101 & 0.2484 & 0.0745 & \bf 0.0231 & \bf 0.0746  & 0.0848  & 0.3\% & 4.5\%  \\ 
 (1000,2,2) & \bf 0.0392 & 0.2834 & 0.1163 & \bf 0.0207 & 0.1178  & \bf 0.1010  & 0.7\% & 5.6\% \\
 (1000,0,4) & \bf 0.0318 & 0.3901 & 0.2308 & \bf 0.0157 & 0.2318  & \bf 0.1679  & 0.4\% & 7.3\% \\
 (1000,16,0)& \bf 0.0564 & 0.3430 & 0.3952 & \bf 0.0937 & 0.3984  & \bf 0.2113  & 0.5\% & 5.5\% \\
 (1000,8,8) & \bf 0.0378 & 0.3637 & 1.2000 & \bf 0.0843 & 1.2014  & \bf 0.2166 & 0.6\% & 6.2\%  \\ 
 (1000,0,16)& \bf 0.0651 & 0.4300 & 1.8348 & \bf 0.0695 & 1.8390  & \bf 0.2544  & 0.9\% & 7.4\% \\ 
 \hline
 (2000,4,0) & \bf 0.0076 & 0.2667 & 0.0365 & \bf 0.0115 & \bf 0.0366  & 0.0826  & 0.5\% & 4.9\% \\  
 (2000,2,2) & \bf 0.0052 & 0.2765 & 0.0547 & \bf 0.0101 & \bf 0.0547  & 0.0866 & 0.6\% & 4.8\%  \\
 (2000,0,4) & \bf 0.0069 & 0.3623 & 0.1263 & \bf 0.0084 & \bf 0.1263  & 0.1397 & 0.6\% & 6.2\%  \\
 (2000,16,0)& \bf 0.0033 & 0.3232 & 0.1986 & \bf 0.0474 & 0.1986  & \bf 0.1519  & 0.5\% & 6.1\% \\
 (2000,8,8) & \bf 0.0035 & 0.3894 & 0.5709 & \bf 0.0436 & 0.5709  & \bf 0.1952 & 0.7\% & 6.9\%  \\ 
 (2000,0,16)& \bf 0.0019 & 0.4321 & 0.8964 & \bf 0.0355 & 0.8964  & \bf 0.2222  & 1.0\% & 8.0\% \\
 \hline
 (4000,4,0) & \bf 0.0051 & 0.2071 & 0.0180 & \bf 0.0060 & \bf 0.0180  & 0.0489 & 0.3\% & 5.5\%  \\ 
 (4000,2,2) & \bf 0.0012 & 0.3031 & 0.0269 & \bf 0.0052 & \bf 0.0269  & 0.0971  & 0.3\% & 5.2\% \\
 (4000,0,4) & \bf 0.0008 & 0.3745 & 0.0593 & \bf 0.0041 & \bf 0.0593  & 0.1444  & 0.3\% & 6.8\% \\
 (4000,16,0)& \bf 0.0027 & 0.3134 & 0.1038 & \bf 0.0233 & \bf 0.1038  & 0.1215 & 0.3\% & 5.7\%  \\
 (4000,8,8) & \bf 0.0036 & 0.3624 & 0.3113 & \bf 0.0219 & 0.3113  & \bf 0.1532  & 0.3\% & 5.4\% \\ 
 (4000,0,16)& \bf 0.0035 & 0.4438 & 0.4753 & \bf 0.0176 & 0.4753  & \bf 0.2156  & 0.3\% & 7.0\% \\ 
 \hline\hline
\end{tabular}
\caption{Compare the Mixed and Cluster-based Designs.}
\label{table: compare with CBR}
\end{table}

From Table~\ref{table: compare with CBR}, we have several key observations. First, let us examine the bias for $\hat\tau_m$ (mixed design using Algorithm~\ref{alg: greedy algorithm}) and $\hat\tau_{cb}$ (cluster-based design). As we can see, the bias of $\hat\tau_m$ is negligible under all test instances, while the bias of $\hat\tau_{cb}$ is significant. This confirms that the mixed randomization design successfully removes the bias inherent in the cluster-based design under our settings.

For empirical variances, Table~\ref{table: compare with CBR} shows that, given $(r_0, r_1)$, if the network size $n$ doubles, the sample variances decrease approximately by half. {Additionally, for fixed $n$, as the expected degree $(r_0 + r_1)$ increases, the variances of $\hat\tau_m$ increase, but not more than $(r_0 + r_1)^2$. These relationships confirm the $O(d^2 / n)$ bound established in Theorem~\ref{theorem: greedy algorihtm upper bound}. We also observe that the empirical variance of $\hat\tau_{cb}$ decreases as its bias increases for a fixed expected degree $r_0 + r_1$, which corroborates Lemma~\ref{lem: cluster-based expectation} and Lemma~\ref{lemma: cluster-based variance}. Larger clusters result in higher variance but fewer between-cluster edges, and vice versa. Although the empirical variance of $\hat\tau_{cb}$ is generally smaller than that of $\hat\tau_m$, the mean squared error (MSE) of $\hat\tau_{cb}$ is not always smaller due to the large bias. As the population size $n$ increases, the mixed design approach becomes more advantageous, demonstrating a decaying MSE. Since $\hat\tau_m$ is consistent, while $\hat\tau_{cb}$ is not, the mixed design approach will significantly outperform the cluster-based design approach under larger populations.}

Finally, we examine the type-1 error of these two estimators under their respective $95\%$ confidence intervals. The type-1 error for $\hat\tau_m$ is significantly less than $5\%$, indicating that the variance estimator in Lemma~\ref{lem: exact upper bound greedy} is indeed conservative. This conservativeness presents challenges for statistical inference: achieving a specific statistical power requires more data than usual. However, in practice, the population size is often much larger than our test instances, which mitigates the influence of the conservative variance estimation.

\subsection{Compare with Other Approaches Under Unknown Interference}

The second task of our numerical experiment is to compare the mixed estimator with other unbiased alternative approaches under unknown interference. The first approach is the pseudo inverse estimator under the first-order interaction model and Bernoulli randomization \citep{cortez2023low-order, eichhorn2024pseudo-inverse}. The first-order interaction model requires the potential outcome function to be a first-order polynomial with unknown coefficients, which aligns with our model in Assumption~\ref{ass: linear model}. The second alternative approach is the Horvitz-Thompson (HT) estimator based on the exposure mapping framework \citep{ugander2023randomized}, which utilizes a randomized cluster algorithm to reduce variance. Notably, this approach only assumes neighborhood interference (i.e., $Y_i$ depends only on the treatments of units in $\{i\} \cup \mathcal{N}_i$) and does not require a specific form for the potential outcome function.

Our mixed randomization design follows Algorithm~\ref{alg: weight-invariant clustering}, where the edge partitioning is generated from the \textit{Neighborhood Expansion} heuristic with a balance factor of 0.3 \citep{zhang2017neighborhood-expansion}. We denote the mixed estimator by $\hat\tau_m$, the pseudo inverse estimator by $\hat\tau_{pi}$, and the HT estimator by $\hat\tau_{ht}$. The test instances are the same as those in Section~\ref{section: numerical comparison known interference}.

We consider three performance metrics: empirical bias, empirical variance, and estimated variance. The empirical bias and variance are calculated from 10,000 independent replications of the experiment. The variance estimation of $\hat\tau_m$ is based on Lemma~\ref{lem: exact upper bound weight invariant}, $\hat\tau_{pi}$ is based on the conservative variance estimator from \cite{aronow2017estimating}, and $\hat\tau_{ht}$ is based on the variance upper bound provided in \cite{ugander2023randomized}. It is important to note that the estimated variances for all three estimators are conservative. We omit the type-1 error result for this numerical simulation since the type-1 error of $\hat\tau_m$ is approximately $0.5\%$, and $0\%$ for both $\hat\tau_{pi}$ and $\hat\tau_{ht}$ across all test instances. The results are summarized in Table~\ref{table:Compare with other approaches}.

\begin{table}[!htbp]
\centering
\begin{tabular}{l|rrr|rrr|rrr} 
 \hline\hline
 \multirow{2}{*}{$(n, r_0, r_1)$} & \multicolumn{3}{c|}{Bias} & \multicolumn{3}{c|}{Var} & \multicolumn{3}{c}{Estimated Var} \\
 \cline{2-10}
  & $\hat{\tau}_m$ & $\hat{\tau}_{ht}$ & $\hat{\tau}_{pi}$ & $\hat{\tau}_m$ & $\hat{\tau}_{ht}$ & $\hat{\tau}_{pi}$ & $\hat{\tau}_m$ & $\hat{\tau}_{ht}$ & $\hat{\tau}_{pi}$ \\  
 \hline
(1000, 4, 0)   & 0.01  & 0.05 & 0.02 & 14.675 & \bf 0.1264 & 0.3478 & 30.128 & \bf 6.6184 & 233.37 \\ 
(1000, 2, 2)   & 0.02 & 0.00 & 0.03 & 10.321 & 0.2929 & \bf 0.2857 & \bf 22.571 & 64.695 & 258.53 \\
(1000, 0, 4)   & 0.01 & 0.03 & 0.01 & 13.468 & 0.5389 & \bf 0.4095 & \bf 23.269 & 83.324 & 378.99 \\
(1000, 16, 0)  & 0.30 & 0.01 & 0.00 & 122.84 & \bf 2.6131 & 2.7423 & 247.01 & \bf 233.80 & 1994.1 \\
(1000, 8, 8)   & 0.08 & 0.07 & 0.01 & 108.77 & 3.7415 & \bf 2.3586 & \bf 212.15 & 6517.1 & 2240.3 \\ 
(1000, 0, 16)  & 0.13 & 0.04 & 0.03 & 158.51 & 5.2616 & \bf 3.2170 & \bf 323.84 & 7343.8 & 2682.0 \\ 
 \hline
(2000, 4, 0)   & 0.16 & 0.01 & 0.00 & 7.3036 & \bf 0.0707 & 0.1565 & 14.831 & \bf 2.9456 & 110.22 \\  
(2000, 2, 2)   & 0.01 & 0.00 & 0.00 & 4.9966 & 0.1548 & \bf 0.1473 & \bf 10.049 & 33.782 & 144.49 \\
(2000, 0, 4)   & 0.02 & 0.00 & 0.05 & 6.9202 & 0.2689 & \bf 0.2106 & \bf 12.525 & 38.062 & 190.01 \\
(2000, 16, 0)  & 0.09 & 0.03 & 0.02 & 63.970 & 1.4521 & \bf 1.4333 & \bf 117.69 & 127.53 & 898.09 \\
(2000, 8, 8)   & 0.13 & 0.00 & 0.01 & 58.479 & 1.8145 & \bf 1.3320 & \bf 115.66 & 6442.4 & 1331.5 \\ 
(2000, 0, 16)  & 0.04 & 0.05 & 0.01 & 78.623 & 2.3043 & \bf 1.6702 & \bf 152.97 & 8527.0 & 1259.2 \\
 \hline
(4000, 4, 0)   & 0.11 & 0.03 & 0.00 & 3.9051 & \bf 0.0391 & 0.0835 & 8.5130 & \bf 1.7293 & 73.573 \\ 
(4000, 2, 2)   & 0.05 & 0.01 & 0.01 & 2.4129 & 0.0794 & \bf 0.0784 & \bf 4.7163 & 22.407 & 65.726 \\
(4000, 0, 4)   & 0.01 & 0.02 & 0.02 & 3.3029 & 0.1296 & \bf 0.1159 & \bf 7.2133 & 23.069 & 71.283 \\
(4000, 16, 0)  & 0.01 & 0.02 & 0.02 & 35.384 & \bf 0.7310 & 0.7751 & 74.246 & \bf 70.175 & 629.40 \\
(4000, 8, 8)   & 0.03 & 0.03 & 0.04 & 28.557 & 0.8977 & \bf 0.6562 & \bf 58.816 & 7401.7 & 423.76 \\ 
(4000, 0, 16)  & 0.04 & 0.05 & 0.02 & 36.871 & 1.1179 & \bf 0.8674 & \bf 73.024 & 7647.2 & 736.26 \\ 
 \hline\hline
\end{tabular}
\caption{Compare with other approaches.}
\label{table:Compare with other approaches}
\end{table}

{We have several observations from the results in Table~\ref{table:Compare with other approaches}. First, the empirical biases of all three estimators are relatively small, consistent with the fact that all three estimators are unbiased. Second, let us compare the variances, which are directly related to the mean squared error (MSE). The weight-invariant design ($\hat\tau_m$) has significantly larger sample variance compared to both the pseudo inverse estimator ($\hat\tau_{pi}$) and the HT estimator ($\hat\tau_{ht}$), indicating substantial room for improvement in developing better clustering algorithms for the mixed design. }

Next, we examine the estimated variances of the three estimators. The mixed estimator $\hat\tau_m$ exhibits the tightest variance estimation among the approaches, whereas the empirical variances of $\hat\tau_{pi}$ and $\hat\tau_{ht}$ are much smaller than their respective estimated variances. This suggests that there might be a gap in the theoretical analysis for these estimators.

Finally, we analyze the impact of the network structure parameters $r_0$ and $r_1$ on the performance of these TTE estimators. Recall that when the average degree $(r_0 + r_1)$ is fixed, a larger $r_0$ implies that more neighbors are located at short distances, while a larger $r_1$ indicates greater influence from long-range neighbors. {Interestingly, we find that the mixed estimator $\hat\tau_m$ and the pseudo inverse estimator $\hat\tau_{pi}$ are robust across all cases, whereas the performance of the HT estimator $\hat\tau_{ht}$ deteriorates as the proportion of long-range neighbors increases. This observation aligns with the theory in \cite{ugander2023randomized}, where it is shown that $\Var(\hat\tau_{ht}) = O\left( \frac{\kappa^4 d^2}{n p (1 - p)} \right)$ under the Randomized Graph Cluster Randomization (RGCR) design, where $\kappa$ represents the graph growth rate. More long-range neighbors lead to a larger $\kappa$ (e.g., when $r_1 = 0$, $\kappa = O(1)$; when $r_0 = 0$, $\kappa \approx d$).}

\section{Conclusions}

We consider the problem of estimating the total treatment effect when units in an experiment may interfere with other units. The interference pattern is represented by an edge-weighted network, where the weights denote the magnitude of interference between a pair of units. {We investigate a novel approach that leverages Bernoulli randomization to eliminate the bias inherent in cluster-based randomization when the network clustering is imperfect.} Building upon previous works \citep{saveski2017detecting, pouget2019hierachicaltesting}, we propose a mixed experimental design that combines both randomization strategies simultaneously.

{For the case when the interference is known, we introduce an unbiased mixed estimator for the total treatment effect. We analyze the variance of this estimator and develop a greedy clustering algorithm to minimize its asymptotic variance. Furthermore, we establish upper and lower bounds on the variance of the mixed estimator under this clustering algorithm, which represent novel contributions to this area of research.}

{For the case when the interference is unknown, we provide a sufficient condition under which the mixed estimator remains unbiased. Additionally, we design an algorithm that meets this condition and guarantees consistent estimation with a mean squared error (MSE) of $O\left(\frac{d^3}{np(1-p)}\right)$. This performance matches the asymptotic behavior of the pseudo inverse estimator \citep{eichhorn2024pseudo-inverse} under a linear model.}

Finally, we discuss the inferential properties of the proposed estimator, showing that it achieves asymptotic normality under certain sparsity conditions in the network structure. These results contribute to the broader understanding of causal inference under network interference and provide practical tools for unbiased and efficient estimation in real-world experimental settings.

\bibliographystyle{unsrtnat}
\bibliography{references}  %%% Uncomment this line and comment out the ``thebibliography'' section below to use the external .bib file (using bibtex) .

\newpage\appendix
\renewcommand{\theequation}{A\arabic{equation}}
\renewcommand{\thelemma}{A.\arabic{lemma}}
\renewcommand{\theproposition}{A.\arabic{proposition}}
\setcounter{equation}{0}

\begin{center}
    \Large \textbf{Appendix}
\end{center}
\section{Proof of Main Results}\label{apx: main}

\subsection{Proof of Lemma~\ref{lemma: cluster-based variance}}\label{apx: proof of lemma cluster-based variance}

\begin{proof} We introduce the following notations:
\begin{align*}
    x_{ij}&:=\boldsymbol{1}\{c(i)=c(j)\},\quad \forall i,j,\\
    t_i&:=\tilde{w}_{i}\left(\frac{z_i}{p}-\frac{1-z_i}{1-p}\right)\\
    r_i&:=\sum_{j\in \mathcal{N}_i}v_{ij}z_i,\quad \forall i.
\end{align*}

Rewrite $\hat\tau_{c}=(\sum_{i=1}^n t_iY_i(\BFz))/(nq)$ and $Y_i(\BFz)=\alpha_i+z_i\beta_i+\gamma r_i$, we have
\begin{equation}
\label{eq:ads}
\begin{split}
        \E[\hat\tau_{c}^2]=&\frac{1}{n^2q^2}E\left(\sum_{1\leq i,j \leq n} t_it_jY_i(\BFz)Y_j(\BFz)  \right)\\
    =&\frac{1}{n^2q^2}\sum_{1\leq i,j \leq n}(\alpha_i\alpha_j\E[t_it_j]+\beta_i\beta_j \E[t_it_jz_iz_j]+\gamma^2 \E[t_it_jr_ir_j]\\&+2\alpha_i\beta_j \E[t_it_jz_j]+2\alpha_i\gamma \E[t_it_jr_j] +2\beta_i\gamma \E[t_it_jz_ir_j] ).
\end{split}
\end{equation}
Notice that $\E[t_it_j]=qx_{ij} \E[t_i^2]$, $\E[t_it_jz_iz_j]=qx_{ij} \E[t_i^2z_i^2]+q^2(1-x_{ij})$, $\E[t_it_jz_j]=qx_{ij} \E[t_i^2z_i]$, $\E[t_it_jr_j]=qx_{ij} \E[t_i^2r_j]$, $\E[t_it_jz_ir_j]=qx_{ij} \E[t_i^2z_ir_j]+q^2(1-x_{ij})\sum_{j'\in \mathcal{N}_j}v_{jj'}x_{jj'}$ and $\E[t_it_jr_ir_j]=qx_{ij}\E[t_i^2r_ir_j]+q^2(1-x_{ij})[(\sum_{i'\in \mathcal{N}_i}v_{ii'}x_{ii'})(\sum_{j'\in \mathcal{N}_j}v_{jj'}x_{jj'})+(\sum_{i'\in \mathcal{N}_i}v_{ii'}x_{ji'})(\sum_{j'\in \mathcal{N}_j}v_{jj'}x_{ij'})]$. Plugging these into Eq~\eqref{eq:ads}, we have
\begin{equation}
\begin{split}
\E[\hat\tau_{c}^2]=&\frac{1}{n^2q^2}\sum_{1\leq i,j \leq n}x_{ij}E\left(t_i^2Y_i(\BFz)(\alpha_j+\beta_jz_i+\gamma r_j)\right)\\
&+\frac{1}{n^2}\sum_{1\leq i,j \leq n}\left(\beta_i\beta_j(1-x_{ij}) +2\beta_i\gamma(1-x_{ij})\sum_{j'\in \mathcal{N}_j}v_{jj'}x_{jj'} \right)\\
&+\frac{\gamma^2}{n^2}\sum_{1\leq i,j \leq n}\left((1-x_{ij})\left[\left(\sum_{i'\in \mathcal{N}_i}v_{ii'}x_{ii'}\right)\left(\sum_{j'\in \mathcal{N}_j}v_{jj'}x_{jj'}\right)+\left(\sum_{i'\in \mathcal{N}_i}v_{ii'}x_{ji'}\right)\left(\sum_{j'\in \mathcal{N}_j}v_{jj'}x_{ij'}\right)\right]\right).
\end{split}
\end{equation}
By Lemma~\ref{lem: cluster-based expectation} and $\Var(\hat\tau_{c})=\E[\hat\tau_{c}^2]-\E[\hat\tau_{c}]^2$, we have
\begin{equation}
\begin{split}\label{apx: lemma 1 proof variance decompose}
\Var(\hat\tau_{c})=&\frac{1}{n^2q^2}\sum_{1\leq i,j \leq n}x_{ij}E\left(t_i^2Y_i(\BFz)(\alpha_j+\beta_jz_i+\gamma r_j)\right)\\
&-\frac{1}{n^2}\sum_{1\leq i,j \leq n}x_{ij}\left(\beta_i\beta_j+2\beta_i\gamma\sum_{j'\in \mathcal{N}_j}v_{jj'}x_{jj'}+\gamma^2(\sum_{i'\in \mathcal{N}_i}v_{ii'}x_{ii'})(\sum_{j'\in \mathcal{N}_j}v_{jj'}x_{jj'}) \right)\\
&+\frac{\gamma^2}{n^2}\sum_{1\leq i,j \leq n}(1-x_{ij})\left(\sum_{i'\in \mathcal{N}_i}v_{ii'}x_{ji'}\right)\left(\sum_{j'\in \mathcal{N}_j}v_{jj'}x_{ij'}\right).\\
\end{split}
\end{equation}

To bound the first term in the right-hand side of Eq~\eqref{apx: lemma 1 proof variance decompose}, recall that $0<Y_L\le Y_i(\BFz)\le Y_M$ for all $i \in V$ and $\BFz \in \{0,1\}^n$, which implies that
$0<Y_L\le \alpha_j+\beta_jz_i+\gamma r_j \le Y_M$ for all $i,j \in V$.
Thus,
\begin{align*}
 \frac{qY_L^2}{p(1-p)}=\E[t_i^2Y_L^2]\le  \E\left[t_i^2Y_i(\BFz)(\alpha_j+\beta_jz_i+\gamma r_j)\right]\le \E[t_i^2Y_M^2]=\frac{qY_M^2}{p(1-p)}, \quad \forall i,j \in V.
\end{align*}

Using the fact that $1/n^2\sum_{1\leq i,j \leq n}x_{ij}=1/n^2\sum_{k=1}^m|C_k|^2=\eta$, we can obtain a bound of 
\[\Theta\left(\frac{\eta}{qp(1-p)} \right)\]

To bound the second term in the right-hand side of Eq~\eqref{apx: lemma 1 proof variance decompose}, note that $Y_i(\boldsymbol{0})>0$ implies $\alpha_i>0$ for all $i$. Let $s_i := \sum_{i'\in \mathcal{N}_i}v_{ii'}x_{ii'}$, then we have
\begin{align*}
& 2Y_M(Y_L-Y_M) \\
\leq &\beta_i\beta_j+\beta_i\gamma s_j+\beta_j\gamma s_i+\gamma^2s_is_j\\
=& (\alpha_i+\beta_i+\gamma s_i)(\alpha_j+\beta_j+\gamma s_j)-2\alpha_i(\alpha_j+\beta_j+\gamma s_j)+\alpha_i\alpha_j\\
\leq & \frac{1}{2}Y_M^2.
\end{align*}
Therefore, the second term is bounded by $O(\eta)$

Finally, we bound the third term in the right-hand side of Eq~\eqref{apx: lemma 1 proof variance decompose}. Let $x_{C_k,i}=x_{i,C_k} :=\boldsymbol{1}\{c(i)=C_k\}$. By Eq \eqref{eq:interference model}, we have
\begin{equation}
\begin{split}\label{eq: delta}
&\sum_{1\leq i,j \leq n}(1-x_{ij})\left(\sum_{i'\in \mathcal{N}_i}v_{ii'}x_{ji'}\right)\left(\sum_{j'\in \mathcal{N}_j}v_{jj'}x_{ij'}\right) \\
=&  \sum_{1\leq k\ne l \leq m}\sum_{i\in C_k}\sum_{j\in C_l}(\sum_{i'\in \mathcal{N}_i}v_{ii'}x_{ji'})(\sum_{j'\in \mathcal{N}_j}v_{jj'}x_{ij'}) \\
=&\sum_{1\leq k\ne l \leq m}\sum_{i\in C_k}(\sum_{i'\in \mathcal{N}_i}v_{ii'}x_{C_l,i'})\sum_{j\in C_l}(\sum_{j'\in \mathcal{N}_j}v_{jj'}x_{C_k,j'}) \\
=&\sum_{1\leq k\ne l \leq m}(\sum_{i\in C_k}\sum_{i'\in \mathcal{N}_i\cap C_l}v_{ii'})(\sum_{j\in C_l}\sum_{j'\in \mathcal{N}_j\cap C_k}v_{jj'})\\
=& n^2\delta. 
\end{split}
\end{equation}
Thus the third term is bounded by $\Theta(\delta)$.
Combine all above together and, we complete the proof.
\end{proof}

\subsection{Proof of Lemma~\ref{lem: size bound}}\label{apx: proof of lemma size bound}
\begin{proof} By Eq~\eqref{eq: delta} and Eq \eqref{eq:interference model}, we have
\begin{align*}
 n^2\delta=   &\sum_{k\ne l}\left(\sum_{i\in C_k}\sum_{i'\in \mathcal{N}_i\cap C_l}v_{ii'}\right)\left(\sum_{j\in C_l}\sum_{j'\in \mathcal{N}_j\cap C_k}v_{jj'}\right)\\
=&\sum_{i,j}\left((1-x_{ij})(\sum_{i'\in \mathcal{N}_i}v_{ii'}x_{ji'})(\sum_{j'\in \mathcal{N}_j}v_{jj'}x_{ij'})\right)\\
\le & \sum_{i,j}(\sum_{i'\in \mathcal{N}_i}|v_{ii'}|x_{ji'})(\sum_{j'\in \mathcal{N}_j}|v_{jj'}|x_{ii'}) \quad \text{(by taking the absolute value)}\\
\le & \sum_{i,j}(\sum_{i'\in \mathcal{N}_i}|v_{ii'}|x_{ji'}) \quad \text{(because  } \sum_{i' \in \mathcal{N}_i}|v_{ii'}| \leq 1)\\
=&\sum_{i=1}^n\sum_{i'\in \mathcal{N}_i}|v_{ii'}||C_{c(i')}|\\
\le &\sum_{i=1}^n (\max_{k\in[m]} |C_k|)\sum_{i'\in \mathcal{N}_i}|v_{ii'}|\\
\le& n\max_{k\in[m]} |C_k|  \quad \text{(because  } \sum_{i' \in \mathcal{N}_i}|v_{ii'}| \leq 1).
\end{align*}
Then the result follows. 
\end{proof}

\subsection{Proof of Theorem \ref{theorem: greedy algorihtm upper bound}}\label{apx: proof of theorem greedy algorihtm upper bound}
\begin{proof}
We would show that under the cluster set $\BFC$ generated by maximum weight matching, the value of the objective function $S$ in Algorithm~\ref{alg: greedy algorithm} is $O(d^2/n)$. Since the merging step in Algorithm~\ref{alg: greedy algorithm} strictly decrease the value of $S(\BFC)$, the final value of $S(\BFC)$ when the algorithm stops is still $S(\BFC)=O(d^2/n)$, which implies the variance of the estimator $\hat\tau_m$ is $O(d^2/(np(1-p)))$ from Eq~\eqref{eq: general variance bound}, Lemma~\ref{lemma: cluster-based variance} and Corollary~\ref{corollary: bernoulli design estimator variance}.

To bound $S(\BFC)$ under the cluster set $\BFC$ generated by maximum weight matching, we first bound $\rho$, $\eta$ and $\tilde\delta$, respectively.
By Theorem~\ref{theorem: maximum weight matching} in Appendix \ref{apx: theorem used}, an initial clustering generated by a maximum weight matching guarantees $\sum_{i\in V}\sum_{j\in \mathcal{N}_i}v_{ij}\boldsymbol{1}\{c(i)=c(j)\} \geq(\sum_{i\in V}\sum_{j\in \mathcal{N}_i}v_{ij})/{2d}$ in a graph with maximum degree $d$, thus
\begin{align*}
    \rho=\frac{\sum_{i\in V}\sum_{j\in \mathcal{N}_i}v_{ij}}{\sum_{i\in V}\sum_{j\in \mathcal{N}_i}v_{ij}\boldsymbol{1}\{c(i)=c(j)\} }\le 2d.
\end{align*}

%By Lemma \ref{lem: function allocation upper bound}, 
To bound $\eta$, because $|C_k|\leq 2$ in the initial clustering and $\sum_{k=1}^m |C_k| = n$, we have
\[
\eta = \sum_{k=1}^m \frac{|C_k|^2 }{ n^2 } \le \sum_{k=1}^m \frac{2|C_k| }{ n^2 } = \frac{2}{n}.
\]

%Let $f(x) = x / n $, where $x \in \{1, \ldots, n\}$. Then, there exists an optimal solution...

By Lemma \ref{lem: size bound}, $\tilde\delta\le 2/n$ since the maximum cluster size is 2. Combining the upper bounds for $\rho$, $\eta$ and $\tilde\delta$, $S(\BFC)\equiv S(\rho,\eta,\tilde\delta)$ defined in \eqref{eq: greedy objective} is $O(d^2/n)$.
\end{proof}

\subsection{Proof of Theorem~\ref{The: lower bound cycle graph}}\label{apx: proof of theorem lower bound cycle graph}

\begin{proof} The proof takes two steps. Firstly, we show that $\Var(\hat\tau_m)=\Omega(\rho^2 \eta)$. In the second step, we show a lower bound of $\rho^2 \eta$ under any cluster design. Note that we have $\alpha_i=1$, $\beta_i=0$ for all $i\in[n]$ and $\gamma=0$, therefore, $E(\hat\tau_b)=\bar\beta=0$ and $\Cov(\hat\tau_c,\hat\tau_b)= E(\hat\tau_c\hat\tau_b)$. Below, we can show $E(\hat\tau_c\hat\tau_b)=0$ since $E\left(\frac{z_i}{p}-\frac{1-z_i}{1-p}\right)=0$ for all $i\in [n]$ and $z_i$, $z_j$ are independent if $c(i)\ne c(j)$.
\begin{align*}
E(\hat\tau_c\hat\tau_b)=&\frac{1}{n^2q(1-q)}\sum_{i,j}E\left(\tilde w_i(1-\tilde w_j)\left(\frac{z_i}{p}-\frac{1-z_i}{1-p}\right)\left(\frac{z_j}{p}-\frac{1-z_j}{1-p}\right)Y_i Y_j \right)\\
=& \frac{1}{n^2}\sum_{i,j}E\left(\left(\frac{z_i}{p}-\frac{1-z_i}{1-p}\right)\left(\frac{z_j}{p}-\frac{1-z_j}{1-p}\right)\Big|\tilde w_i=1,\tilde w_j=0 \right)\\
=0.
\end{align*}
Therefore, according to the bound of Lemma~\ref{lemma: cluster-based variance}, we have
\begin{align*}
    \Var(\hat\tau_m)=&\rho^2\Var(\hat\tau_c)+(1-\rho)^2\Var(\hat\tau_b)
    \ge\rho^2\Var(\hat\tau_c)= \Omega\left(\frac{\rho^2 \eta}{p(1-p)}\right)
\end{align*}

We now prove the lower bound of $\rho^2\eta$. In a $(d,\kappa)$-cycle network, we call an edge $(i,j) \in E$ a Type-1 edge if $j\in\{(i\pm 1)\mod n, (i\pm 2)\mod n,...,(i\pm (\kappa-1))\mod n \}$, and a Type-2 edge if $j\in \{(i\pm \kappa)\mod n, (i\pm 2\kappa)\mod n,...,(i\pm d\kappa)\mod n \}$. We say a cluster $\{i_1,i_2,...,i_l\}$ is continuous if and only if 
it includes at least $l-1$ edges from the set $\{(j,(j+1)\mod n):1\le j\le n\}$. Let $E^1$ and $E^2$ be the set of Type-1 and Type-2 edges, respectively. 

Consider a cluster $C$ of size $t$. Let $n_k^C$ denote the number of Type-$k$ edges ($k=1,2$) in cluster $C$, i.e. $n_k^C=|\{(i,j) \in E: i \in C, j\in C;\; (i,j)\in E^k\}|$, for $k\in\{1,2\}$. First, we bound $n_1^C$ as follows:
\begin{enumerate}
\item[(1-a)] When $t\le \kappa$, since the maximum number of edges in $C$ is $t(t-1)$ (when $C$ is a clique), we have
$n_1^C \le  t(t-1)$.

\item[(1-b)] When $\kappa+1\le t\le 2\kappa-1$, 
\begin{align*}
    n_1^C\le (\kappa-1)+\kappa+...+\underbrace{(t-1)+...+(t-1)}_{2\kappa-t}+(t-2)+...+(\kappa-1)=(\kappa-1)(2t-\kappa).
\end{align*}

\item[(1-c)] When $t\ge 2\kappa$, 
\begin{align*}
    n_1^C\le (\kappa-1)+\kappa+...+\underbrace{2(\kappa-1)+...2(\kappa-1)}_{t-2\kappa+2}+2\kappa-3+...+(\kappa-1)=(\kappa-1)(2t-\kappa).
\end{align*}
\end{enumerate}

% All the equality holds when $C$ is continuous and $n\ge t+2\kappa-2$. 

Similarly, for $n^C_2$, we have the following bounds:
\begin{enumerate}
    \item[(2-a)] When $t\le (d+1)\kappa$, we have $n_2^C\le t\lfloor \frac{t-1}{\kappa}\rfloor$.
    \item[(2-b)] When $(d+1)\kappa+1\le t\le 2d\kappa+1$,
\begin{align*}
    n_2^C\le& d\kappa+(d+1)\kappa+...+ \lfloor\frac{t-1}{\kappa}\rfloor \left(t-2\kappa(\lfloor\frac{t-1}{\kappa}\rfloor-d)\right)+(\lfloor\frac{t-1}{\kappa}\rfloor-1)\kappa+...+d\kappa\\
    =& t\lfloor\frac{t-1}{\kappa}\rfloor-(\lfloor\frac{t-1}{\kappa}\rfloor-d)(\lfloor\frac{t-1}{\kappa}\rfloor-d+1)\kappa.
\end{align*}
\item[(2-c)] When $t\ge 2d\kappa+2$,
\begin{align*}
    n_2^C\le& d\kappa+(d+1)\kappa+...+2d\left(t-2d\kappa\right)+(2d-1)\kappa+... +d\kappa=d(2t-\kappa(d+1)).
\end{align*}
\end{enumerate}

% All the equalities hold when $t\mod \kappa=0$ and $n\ge t+2d\kappa+1$.

Note that when $(d+1)\kappa+1\le t\le 2d\kappa+1$ (i.e., Case 2-b),
\begin{align*}
    &t\lfloor\frac{t-1}{\kappa}\rfloor-(\lfloor\frac{t-1}{\kappa}\rfloor-d)(\lfloor\frac{t-1}{\kappa}\rfloor-d+1)\kappa\\
    \le & k(\lfloor\frac{t-1}{\kappa}\rfloor+1)\lfloor\frac{t-1}{\kappa}\rfloor-(\lfloor\frac{t-1}{\kappa}\rfloor-d)(\lfloor\frac{t-1}{\kappa}\rfloor-d+1)\kappa\\
    = &\kappa d(2\lfloor\frac{t-1}{\kappa}\rfloor-d+1)\\
    \le& d(2t-\kappa(d+1)),
\end{align*}
so the bound in (2-c) also applies to (2-b).

To summarize, 
\begin{align*}
 n_1^C+n_2^C\le f(t)\ :=\ \left\{
\begin{array}{cl}
t(t-1), &  \text{if } 1\le t\le \kappa\\
 (\kappa-1)(2t-\kappa)+t\lfloor\frac{t-1}{\kappa}\rfloor, &  \text{if } \kappa+1\le t\le \kappa(d+1) \\
 d(2t-\kappa(d+1))+(\kappa-1)(2t-\kappa), &  \text{if }  t\ge \kappa(d+1)+1. \\
\end{array} \right.
\end{align*}

It's easy to check that $f(t)$ is nonnegative and increasing on $t\ge 1$. 
Suppose we have $m$ clusters, and $t_j$ is the size of the $j$th cluster. Then by $n_1^C+n_2^C\le f(t)$ and the Cauchy-Schwarz inequality,
\begin{align*}
\rho^2\eta\ge\frac{\sum_{j=1}^m t_j^2}{n^2}\left(\frac{2n(d+\kappa-1)}{\sum_{i=j}^m f(t_j)}\right)^2\ge 4(d+\kappa-1)^2 \left(\sum_{j=1}^m \frac{f(t_j)^2}{t_j^2}\right)^{-1}.
\end{align*}

Furthermore, we relax $t_j$ to continuous variables.
Note that when $1\le t\le \kappa$, $f(t)/t^{1.5}\le t^{0.5}\le \sqrt{\kappa}$. When $\kappa+1\le t\le \kappa(d+1)$, we have
\begin{align*}
    f(t)/t^{1.5}\le 2\kappa/\sqrt{t}+\sqrt{t}/\kappa\le 2\kappa/\sqrt{\kappa+1}+\sqrt{(d+1)/\kappa}.
\end{align*}

When $\kappa(d+1)+1\le t\le n$, $f(t)/t^{1.5}\le 2(d+\kappa-1)/\sqrt{t}< 2(d+\kappa-1)/\sqrt{\kappa(d+1)}$. Thus, $f(t)^2/t^3=O(\max\{\kappa,d/\kappa\})$. Then $ \sum_{i=j}^m {f(t_j)^2}/{t_j^2}=\sum_{i=j}^m t_j{f(t_j)^2}/{t_j^3}=O(\max\{n\kappa,nd/\kappa\}\sum_{i=j}^m t_j)= O(\max\{n\kappa,nd/\kappa\})$. Therefore, $\rho^2\eta=\Omega(\min\{d^2/(n\kappa),d\kappa/n\})$ and the proof is complete.
\end{proof}

\subsection{Proof of Theorem~\ref{theorem: the algorithm is weight invariant}}\label{apx: proof of theorem algorithm is weight invariant}

\begin{proof} By Definition~\ref{def: proper edge partitioning}, the edge partitioning $E=E_1\cup...\cup E_u$ is proper implies that $E_k=\{(i,j)\in E| i\in V(E_k),\; j\in V(E_k)\}$, for all $k\in[u]$. Therefore, if $(i,j)\in E$, then there is only one element in the set $\{E_k| i\in V(E_k),\; j\in V(E_k)\}$, otherwise $E_1,...,E_u$ are not mutually disjoint. Suppose $(i,j)\in E_k$, by Algorithm~\ref{alg: weight-invariant clustering}, we have
\[P(c(i)=c(j))= P(V(E_k) \text{ is assigned as a cluster})=P(X_k< X_l,\;\forall l\ne k,\; M_{kl}=1) \]

Since $X_1,X_2, \ldots,X_u$ are independent Exponential distributed random variables with parameter $\omega_1,...,\omega_u$, and $\omega$ is the eigenvector of $M$ associated with the largest eigenvalue $\lambda*$, we have 
\begin{align*}
    P(X_k< X_l,\;\forall l\ne k,\; M_{kl}=1)=\frac{\omega_k}{\sum_{l: M_{kl}=1} \omega_l}=\frac{\omega_k}{(M\omega)_k}=1/\lambda*,\;\forall k\in [u].
\end{align*}

Note that the incidence matrix $M$ is nonnegative and irreducible since $M$ can be associated with a strongly connected graph. According to the Perron–Frobenius theorem,
the largest eigenvalue $\lambda^*$ of $M$ will also be the Perron–Frobenius eigenvalue, which satisfies 
\begin{align}\label{eq: Perron–Frobenius}
    \min_i \sum_{j\in [u]}M_{ij}\le \lambda^*\le \max_i \sum_{j\in [u]}M_{ij}
\end{align}

Therefore, $\lambda^*\ge 1$. The proof is completed by combining all results together and apply Lemma~\ref{lemma: weight invariant design}.
\end{proof}

\subsection{Proof of Theorem~\ref{theorem: weight invariant variance}}\label{apx: proof of theorem weight invariant variance}

\begin{proof} According to the variance decomposition formula,
\[ \Var(\hat\tau_m)=\Var(\E[\hat\tau_m|\BFC])+\E[\Var(\hat\tau_m|\BFC)], \]
where $\BFC=\{C_1,...,C_m\}$ is the clustering generated by Algorithm~\ref{alg: weight-invariant clustering}.
Our proof takes two steps. Firstly, we bound $\E[\Var(\hat\tau_m|\BFC)]$. Secondly, we bound $\Var(\E[\hat\tau_m|\BFC])$. To bound $\E[\Var(\hat\tau_m|\BFC)]$, we apply Lemma~\ref{lemma: cluster-based variance}, Corollary~\ref{corollary: bernoulli design estimator variance} and Eq~\eqref{eq: general variance bound}, which tell
\[\Var(\hat\tau_m|\BFC)= O\left( \frac{\rho^2\eta}{p(1-p)}+\rho^2\delta \right)\]
given that $q$ is a fix constant. According to our assumption $|V(E_i)|=O(1)$, $\forall i\in[u]$, Algorithm~\ref{alg: weight-invariant clustering} will generate $\BFC$ with $C_i=O(1)$, $\forall i\in[m]$. Therefore, $\eta=n^{-2}\sum_{i=1}^m |C_i|=O(n^{-1})$. Similarly, we have $\delta= O(n^{-1})$ by Lemma~\ref{lem: size bound}. Since $\rho=\lambda^*$, we use Eq~\eqref{eq: Perron–Frobenius} to bound $\rho$. Notice that $|V(E_i)|=O(1)$, $\forall i\in[u]$ and the maximum degree of the network is $d$, then we have
\begin{align}
    \sum_{j\in [u]}M_{ij}=\sum_{j\in [u]}\boldsymbol{1}\{V(E_j)\cap V(E_i)\ne \emptyset\}=O(d)\;\forall i\in [u]. \label{eq: bound for rho and lambda}
\end{align}

Thus, $\rho=\lambda^*\le \max_i \sum_{j\in [u]}M_{ij} = O(d)$. Combine the bounds for $\rho$, $\eta$ and $\delta$ together, we have
\[\E[Var(\hat\tau_m|\BFC)]= O\left( \frac{d^2}{np(1-p)} \right)\]

Next, we bound $\Var(\E[\hat\tau_m|\BFC])$. Analogues to the proof of Lemma~\ref{lem: cluster-based expectation}, we have $\E[\hat\tau_m|\BFC]=\hat\beta+\rho \gamma/n\sum_{i=1}^n\sum_{j\in \mathcal{N}_i}v_{ij}\boldsymbol{1}\{c(i)= c(j)\}$. To simplify the notation, let $x_{ij}=\boldsymbol{1}\{c(i)= c(j)\}$. Then we have
\begin{align*}
    \Var(\E[\hat\tau_m|\BFC])=&\Var\left(\bar\beta+\rho \gamma/n \sum_{i=1}^n\sum_{j\in \mathcal{N}_i}v_{ij}x_{ij}\right)
    =\frac{\rho^2\gamma^2}{n^2}\Var\left(\sum_{(i,j)\in E}v_{ij}x_{ij}\right).
\end{align*}

To bound the last term in the above equations, we study the value of $\Cov(x_{ij},x_{kl})$ under two scenarios. 
\begin{enumerate}
    \item Suppose $(i,j)\in E_g$ and $(k,l)\in E_h$ if there exist $f\in [u]$ such that $M_{fg}=M_{fh}=1$, then according to Algorithm~\ref{alg: weight-invariant clustering}, $x_{ij}$ and $x_{kl}$ are not independent. In this case, we can use the bound $|\Cov(x_{ij},x_{kl})|=|\E[x_{ij}x_{kl}]-\E[x_{ij}]\E[x_{kl}]|=|\E[x_{ij}x_{kl}]-1/\rho^2|\le 1/\rho$.
    \item Otherwise, $x_{ij}$ and $x_{kl}$ are independent, and $\Cov(x_{ij},x_{kl})=0$.
\end{enumerate}

Given $(i,j)\in E_g$, $g\in [u]$, we are interested in the number of possible $(k,l)\in E$ such that $\Cov(x_{ij},x_{kl})\ne 0$
\begin{align*}
    \sum_{(k,l)\in E}\boldsymbol{1}\{\Cov(x_{ij},x_{kl})\ne 0\}\le & 
\sum_{h\in [u]}\sum_{(k,l)\in E_h} \boldsymbol{1}\left\{\sum_{f\in [u]} M_{fg}M_{fh}>0\right\} \\
\le & \max_{i\in [u]} |E_i|\sum_{h\in [u]}\sum_{f\in [u]} M_{fg}M_{fh}\\
\le & \max_{i\in [u]} |V(E_i)|^2\sum_{f\in [u]} M_{fg}\sum_{h\in [u]}M_{fh}\\
=& O(d^2),
\end{align*}
where the equality is due to the bounds from Eq~\eqref{eq: bound for rho and lambda} and $|V(E_i)|=O(1)$. Combining above bound and $|\Cov(x_{ij},x_{kl})|\le 1/\rho$, $\sum_{j=1}^n|v_{ij}|\le 1$, $\forall i\in [i]$, we rewrite
\begin{align*}
\Var\left(\sum_{(i,j)\in E}v_{ij}x_{ij}\right)=&\sum_{(i,j)\in E}\sum_{(k,l)\in E}v_{ij}v_{kl}\Cov(x_{ij},x_{kl})\\
=&\sum_{(i,j)\in E}\sum_{(k,l)\in E}v_{ij}v_{kl}\Cov(x_{ij},x_{kl})\boldsymbol{1}\{\Cov(x_{ij},x_{kl})\ne 0\}\\
\le & \sum_{(i,j)\in E}\sum_{(k,l)\in E}|v_{ij}||v_{kl}||\Cov(x_{ij},x_{kl})|\boldsymbol{1}\{\Cov(x_{ij},x_{kl})\ne 0\}\\
\le &\sum_{(i,j)\in E}|v_{ij}|O(d^2/\rho)\\
=& O(nd^2/\rho).
\end{align*}
Therefore, we can bound $\Var(\E[\hat\tau_m|\BFC])= O(d^2\rho/n)= O(d^3/n)$, since $\rho=O(d)$ in our previous analysis. Combining all results together, the proof is completed.
\end{proof}

\subsection{Proof of Theorem~\ref{thm:asymptotic_normality}}\label{apx: proof of asymptotic_normality}

\begin{proof} Let $L_i$ be the term within the summation in Eq~\eqref{eq: estimator reformulate for clt}. and $X_i=\frac{1}{\sqrt{n} \rho} (L_i-\E[L_i])$, then $\E[X_i]=0$, $\E[|X_i|^3]\le \left(\frac{2}{\sqrt{n} p(1-p)}Y_M \right)^3$, and $\E[X_i^4]\le \left(\frac{2}{\sqrt{n} p(1-p)}Y_M \right)^4<\infty$. According to the assumption, $\liminf_{n\rightarrow\infty} \sigma^2_n=\Var(\sum_{i=1}^n X_i)>0$.

Firstly, consider the fixed clustering (Algorithm~\ref{alg: greedy algorithm}). For two units $i$ and $j$, if $\{c(k)|k\in \mathcal{N}_i\}\cap \{c(l)|l\in \mathcal{N}_j\}= \emptyset$, then $X_i$ and $X_j$ are independent under Assumption~\ref{ass: linear model}. Note that $|\{c(k)|k\in \mathcal{N}_i\}|$ is at most $d+1$ for all $i\in [n]$ and each cluster is connected to at most $O(d)$ nodes under condition (1). Using the notation $D$ in Lemma~\ref{lem: wasserstein bound}, $D=\max_i\{j| X_i \text{ and } X_j \text{ are dependent}\}=O(d^2)$. Apply Lemma \ref{lem: wasserstein bound} and the result follows.

Under condition (2), $X_i$ and $X_j$ are independent if $\{k\in [u]|\mathcal{N}_i\cap V(E_k)\ne\emptyset\;,\mathcal{N}_j \cap V(E_k)\ne\emptyset\}= \emptyset$. Note that $\mathcal{N}_i$ is $O(d)$, $\forall i\in [n]$ and each node is connected to at most $O(d)$ edge partitions (Eq~\eqref{eq: bound for rho and lambda}). Therefore, $D=\max_i\{j| X_i \text{ and } X_j \text{ are dependent}\}=O(d^3)$. Apply Lemma \ref{lem: wasserstein bound} and the result follows.
\end{proof}

\subsection{Theorem used}\label{apx: theorem used}
\begin{theorem}\label{theorem: maximum weight matching}
For a graph $G(V, E)$ with the edge weights $\{v_{ij}\}_{(i,j) \in E}$ and the maximum degree $d$, the sum of weights in a maximum weight matching is at least $(\sum_{i\in V}\sum_{j\in \mathcal{N}_i}v_{ij})/{2d}$. 
\end{theorem}

\begin{proof} 
    Let $G'(V')$ be the subgraph of $G(V,E)$ induced by a subset of vertices $V'\subset V$ and $G'(V',E')$ the subgraph induced by $V'\subset V$ and $E'\subset E$. For brevity of notations, we  Let $V(E)$ denote the set of all vertices incident to edges in $E$. We provide Algorithm \ref{alg: Graph decomposition into matchings} based on maximum matching.
    
\begin{algorithm}[!ht]\label{alg: Graph decomposition into matchings}
\DontPrintSemicolon
\KwInput{A graph $G(V,E)$}
    Initialize the clustering set $\mathbb{C}\leftarrow\{\}$\;
    Initialize set $A\leftarrow E$\;
   \While{$|A|>0$}
   {
   $U\leftarrow\{u\in V|\texttt{degree}(u,A)=\max_{v\in V}\texttt{degree}(v,A)\}$\;
        Find a maximum matching $M_1=\{(v_1,u_1),(v_2,u_2),...\}$ from the subgraph $G'(U)$\;
        $M_2\leftarrow \{\}$\;
        \For{$u\in U$}{
            \If{$u\in U/V(M_1)$ and $\mathcal{N}_u\not\subset U$}{
                Choose one $u'\in \mathcal{N}_u/ U$ and let $M_2\leftarrow M_2\cup\{(u,u')\}$
            }
        }
        $\mathbb{C}\leftarrow \mathbb{C}\cup\{ \texttt{edge\_set\_to\_clustering}(M_1\cup M_2)\}$\;
        $A\leftarrow A/(M_1\cup M_2)$\;
        $U'=U/V(M_1\cup M_2)$\;
        \If{$|U'|>0$}{
         Find a maximum matching $M_3$ from the subgraph $G'(U,\{(u',v')\in A|u'\in U',v'\in U\})$\;
         $\mathbb{C}\leftarrow \mathbb{C}\cup\{ \texttt{edge\_set\_to\_clustering}(M_3)\}$\;
         $A\leftarrow A/M_3$\;
        }
   }
  \KwOutput{Randomly return one element from $\mathbb{C}$ with probability $|\mathbb{C}|^{-1}$}
\caption{Graph decomposition into matchings}
\end{algorithm}

The function $\texttt{edge\_set\_to\_clustering}(E)$ turns an edge set $E$ into a clustering design. If $(u,v)\in E$, then $u$ and $v$ are in the same cluster. If $u$ is not covered by $E$, the $u$ itself is a cluster. The idea of this algorithm is to decompose the graph $G(V,E)$ into clustering set $\mathbb{C}=\{\BFC_1,\BFC_2,...\}$, such that each edge is covered by exactly one cluster $C^k_j$ from a specific clustering $\BFC_k$. 

The set $A$ represents the uncovered edges. In line 4, we found a vertex set $U$ with the largest degree. In line 5, we found a maximum matching $M_1$ from the subgraph $G'(U)$, thus an unmatched vertex in $U$ cannot be adjacent to another unmatched vertex. In lines 6 to 9, we found edge set $M_2$ such that each vertex in $V(M_1\cup M_2)\cap U$ has exactly one edge from $M_1\cap M_2$ incident to it. Also, the edges covered by clustering $\texttt{edge\_set\_to\_clustering}(M_1\cup M_2)$ is exactly $M_1\cup M_2$. After constructing a clustering from $M_1$ and $M_2$ in line 10, we found the set of vertex $U'$ without incident to $M_1\cup M_2$. Since the graph $G'(U,\{(u',v')\in A|u'\in U',v'\in U\})$ is bipartite and every $u'\in U'$ has the same degree, the maximum matching $M_3$ we found in line 14 is always a perfect matching by Hall's marriage theorem. Thus, remove $M_1$ ,$M_2$ and $M_3$ from $A$ will strictly decrease $\max_{v\in V}\texttt{degree}(v,A)$. Since the maximum degree of graph $G(V,E)$ is $d$, the algorithm will terminate after at most $d$ while loops. Also, once an edge is covered by a clustering, we directly remove it from the uncovered edge set $A$, making each edge is covered exactly once.

Algorithm \ref{alg: Graph decomposition into matchings} decomposes a graph with maximum degree $d$ to at most $2d$ graph matchings, such that each edge in $E$ is covered exactly once. This implies that a maximum weight matching under the weights $\{v_{i,j}\}_{i,j\in [n], i\ne j}$ is at least $(\sum_{i\in V}\sum_{j\in \mathcal{N}_i}v_{ij})/{2d}$.
\end{proof}

\section{Additional Numerical Results}\label{apx: additional numerical result}

{Here we provide a numerical counterexample to show that the regression adjustment approach is biased under our setting, which is discussed in Section~\ref{sec: compare with regression adjustment}. The test instances are generated in the same way as in Section~\ref{section: test instances} except that we multiply each $\hat\alpha_i$ by the degree of node $i$ for all $i\in[n]$. The network size $n$ is chosen from $\{1000,2000,4000\}$. The value of $(r_0,r_1)$ is chosen from the set $\{(4,0),(8,0)\}$.
We use $\hat\tau_m$ and $\hat\tau_{ra}$ to denote the mixed and regression adjustment estimator, respectively. We examine the empirical bias and variance of these two estimators under 10,000 times simulation. The results are summarized as follows:}

\begin{table}[!htbp]
\centering
\begin{tabular}{c|c|c|c|c}
\hline\hline
$(n, r_0, r_1)$ & Bias($\hat\tau_m$) & Bias($\hat\tau_{ra}$) & $\Var(\hat\tau_m)$ & $\Var(\hat\tau_{ra})$ \\ 
\hline
(1000, 4, 0) &\bf 0.014 & 0.723 & 0.463 & \bf 0.009 \\ 
(1000, 8, 0) & \bf 0.005 & 1.517 & 3.024 &\bf  0.033 \\ 
(2000, 4, 0) &\bf  0.001 & 0.831 & 0.332 &\bf  0.005 \\ 
(2000, 8, 0) &\bf  0.003 & 1.696 & 2.128 &\bf  0.017 \\ 
(4000, 4, 0) &\bf  0.001 & 0.933 & 0.192 &\bf  0.004 \\ 
(4000, 8, 0) & \bf 0.002 & 1.984 & 1.210 &\bf  0.011 \\
\hline\hline
\end{tabular}
\caption{Compare the mixed and regression adjustment estimators.}
\label{table: compare the mixed and regression adjustment estimators}
\end{table}

We can observe that the empirical bias of the mixed estimator $\hat\tau_m$ is negligible, while the regression adjustment estimator $\hat\tau_{ra}$'s is very severe. Although the empirical variance of the regression adjustment is several times smaller in magnitude compared with the mixed estimator, its mean squared error will not converge to zero as the population size grows.

\end{document}